\documentclass[3p,12pt]{elsarticle}

\usepackage{lineno,hyperref}
\usepackage{amsmath,bm,geometry}
\usepackage{amsfonts}
\usepackage{color}
\usepackage{nomencl,etoolbox,ragged2e,siunitx}
\usepackage{mathtools}
\usepackage{multirow}
\usepackage{caption}
\usepackage{parskip}
\usepackage{calligra}
\usepackage{makeidx}
\usepackage{tikz}
\usepackage[utf8]{inputenc}
\usepackage[english]{babel}
\usepackage{amsthm}
\usepackage{cleveref}
\newtheorem{theorem}{Theorem}
\usetikzlibrary{matrix}
\makeindex
\DeclareMathAlphabet{\mathcalligra}{T1}{calligra}{m}{n}
\begin{document}
\title{A Unified Gas-kinetic Scheme for Micro Flow Simulation Based on Linearized Kinetic Equation}
\author[ad1]{Chang Liu}
\ead{cliuaa@connect.ust.hk}
\author[ad1,ad2]{Kun Xu\corref{cor1}}
\ead{makxu@ust.hk}
\address[ad1]{Department of Mathematics, Hong Kong University of Science and Technology, Hong Kong}
\address[ad2]{Department of Mechanical and Aerospace Engineering, Hong Kong University of Science and Technology, Hong Kong, China}
\cortext[cor1]{Corresponding author}

\begin{abstract}
The flow regime of micro flow varies from collisionless regime to hydrodynamic regime according to the Knudsen number.
On the kinetic scale, the dynamics of micro flow can be described by the linearized kinetic equation.
In the continuum regime, hydrodynamic equations such as linearized Navier-Stokes equations and Euler equations can be derived from the linearized kinetic equation by the Chapman-Enskog asymptotic analysis.
In this paper, based on the linearized kinetic equation we are going to propose a unified gas kinetic scheme scheme (UGKS) for micro flow simulation, which is an effective multiscale scheme in the whole micro flow regime.
The important methodology of UGKS is the following. Firstly, the evolution of microscopic distribution function is coupled with the evolution of macroscopic flow quantities. Secondly, the numerical flux of UGKS is constructed based on the integral solution of kinetic equation, which provides a genuinely multiscale and multidimensional numerical flux. The UGKS recovers the linear kinetic solution in the rarefied regime, and converges to the linear hydrodynamic solution in the continuum regime. An outstanding feature of UGKS is its capability of capturing the accurate viscous solution even when the cell size is much larger than the kinetic kinetic length scale, such as the capturing  of the viscous boundary layer with a cell size ten times larger than the particle mean free path. Such a multiscale property is called unified preserving (UP) which has been studied in \cite{guo2019unified}. In this paper, we are also going to give a mathematical proof for the UP property of UGKS.
\end{abstract}

\maketitle

\section{Introduction}
The Boltzmann equation is a fundamental equation in kinetic theory, and it is widely applied in the fields of aerospace engineering, chemical industry, as well as the microelectromechanical systems (MEMS). The modeling scale of the Boltzmann equation is on the kinetic scale, namely the
particle mean free path and collision time scale. Such small modeling scale makes the Boltzmann equation reliable but on the other hand quite complicated, comparing to the hydrodynamic scale Navier-Stokes (NS) and Euler equtions. The complication of the Boltzmann equation comes from its high dimension and stiff collision operator. The asymptotic theories, such as the Hilbert expansion and Chapman-Enskog theory, have been developed that bridge the Boltzmann equation and the hydrodynamic equations, and connect the kinetic parameters to the hydrodynamic ones \cite{chapman1990mathematical,cercignani1969mathematical}.
Similar to the asymptotic theories, the linearized Boltzmann equation has also been studied when dealing with the small perturbed flow field in MEMS and porous media.
For such small perturbed flow field, the linearized Boltzmann equation can faithfully recover the physical solution in a much effective way \cite{cercignani1969mathematical}. As shown in Fig.\ref{diagram}, the asymptotic analysis can also be applied on the linearized kinetic equation, which means the linearized kinetic equation can be approximated by the linearized NS and Euler equations in the hydrodynamic scale. In many applications, the flow regime or the local Knudsen number varies several order of magnitude in a single computation,
and therefore an effective multiscale numerical scheme is highly demanded.
\begin{figure}
\centering
\includegraphics[width=0.9\textwidth]{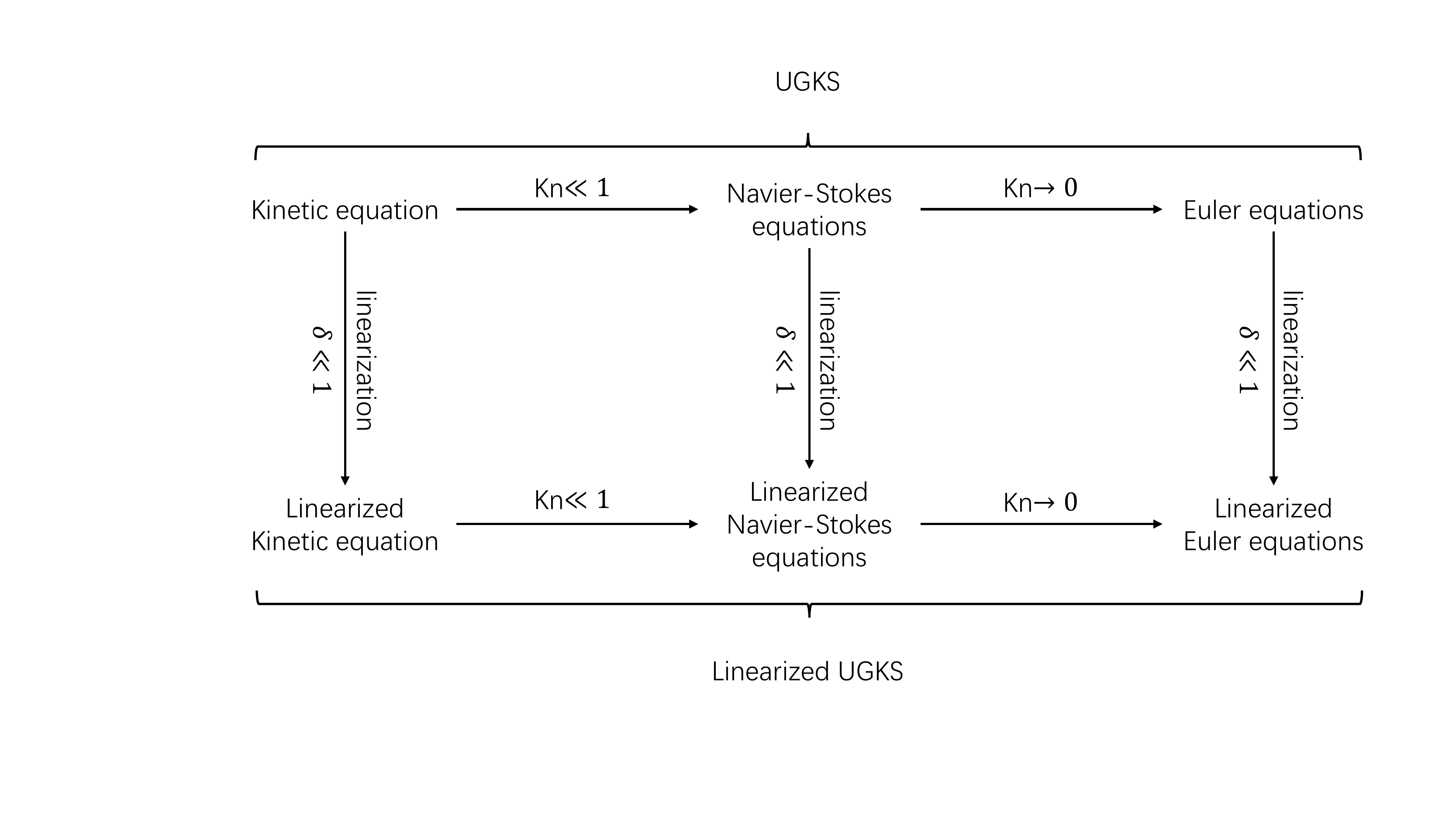}
\caption{A diagram of governing equations according to flow regime, and corresponding multiscale numerical scheme.}
\label{diagram}
\end{figure}

For the last several decades, researchers has been trying to develop effective multiscale numerical schemes \cite{ugks2010,guo2013discrete,su2020can,yuan2019multi,jenny2010solution,fei2020unified}.
In 2010, Xu et al. proposes the unified gas-kinetic scheme, which is the first genuine multiscale scheme being able to capture the viscous effect with cell size much larger than the kinetic scale. The direct modeling methodology of UGKS are: firstly, the evolution of microscopic distribution function is coupled with the evolution of macroscopic flow quantities; secondly, the numerical flux of UGKS is constructed based on the evolution  solution of kinetic equation, which provides a genuinely multiscale numerical flux \cite{xu-book}. The UGKS has been successfully applied in radiative transfer \cite{sun2015asymptotic,sun2017,sun2018,li2020unified}, plasma transport \cite{liu2017}, and multiphase flow \cite{liu2019unified}, etc. In order to reduce the computational cost, the unified gas-kinetic wave-particle (UGKWP) method, i.e. a stochastic version of UGKS has been proposed by Liu et al. \cite{liu2020unified}. The UGKWP method is a multiscale and asymptotic complexity diminishing scheme \cite{crestetto2019asymptotically}, which is effective in the simulation of three dimensional hypersonic flow problems in all regimes. The discrete unified gas-kinetic scheme (DUGKS) is developed by Guo el al. is also a multiscale scheme \cite{guo2013discrete,guo2015discrete}, and has been successfully applied in the field of micro flow \cite{zhu2019application,wang2018oscillatory}, gas mixture \cite{zhang2018discrete}, gas-particle multiphase flow \cite{tao2018combined}, phonon transport \cite{zhang2017unified}, radiation \cite{luo2018multiscale}, etc. The general synthetic iteration scheme was first proposed by Wu et al. for the steady state solution of the linearized kinetic eqaution \cite{su2020can}, and is recently extended to the simulation of nonlinear kinetic equation and diatomic gas \cite{su2020fast,zhu2020general}.

In order to measure the capability of numerical schemes in capturing the multiscale flow physics, Guo et al. proposes the concept of unified preserving property (UP) \cite{guo2019unified}. The unified preserving property states that if a scheme is able to preserve the $n$-th order Chapman-Enskog expansion in continuum regime, with $\Delta t<\text{Kn}^\alpha$, then the scheme is an $n$-th order UP scheme. Therefore the
UP capability of a scheme can be measured by two parameters $n$ and $\alpha$. A higher $n$ and lower $\alpha$ indicates a better multiscale scheme.

In this paper, we extend UGKS to the micro flow simulation and propose a linearized version of UGKS.
The advantage of the linearized UGKS compared to the nonlinear UGKS is that it is much more efficient and accurate in the micro flow simulation. However, the drawback is that the application of linearized version of UGKS is limited to the small perturbed flow problems.
The rest of this paper is organised as following. In Section 2, we are going to specify the linearized kinetic equation and propose the unified gas-kinetic scheme for linearized system. In Section 3, we are going to analyze the unified preserving property of UGKS. The numerical tests are shown in Section 4. And Section 5 will be the conclusion.

\section{Unified Gas-kinetic Scheme for Linearized Kinetic Equation}
\subsection{Linearized kinetic equation}
In this work, the kinetic BGK equation is considered \cite{BGK1954},
\begin{equation}\label{BGK}
\frac{\partial f}{\partial t}+ \vec{v}\cdot\nabla_{\vec{x}} f = \frac{g-f}{\tau},
\end{equation}
where $f(\vec{x},t,\vec{v})$ is the velocity distribution function of gas particle, $\tau$ is the local relaxation parameter which is determined by $\tau=\mu/p$
with the gas pressure $p$ and dynamic viscosity $\mu$.
The local equilibrium Maxwellian distribution $g(\vec{x},t,\vec{v})$ has the form
\begin{equation}\label{maxwellian}
  g(\vec{x},t,\vec{v})=\rho\left(\frac{m}{2\pi k_BT}\right)^{-\frac{3}{2}} \exp\left(-\frac{m(\vec{v}-\vec{U})^2}{2k_BT}\right),
\end{equation}
with density $\rho$, velocity $\vec{U}$, temperature $T$, Boltzmann constant $k_B$, molecular mass $m$.
For the study of micro flow, assume that the unperturbed velocity $\vec{U}_0=0$, and the following dimensionless is used,
\begin{equation}
\tilde{\rho}=\frac{\rho}{\rho_0},
\quad \tilde{T}=\frac{T}{T_0},
\quad \tilde{\vec{U}}=\frac{\vec{U}}{|\vec{U}_{ref}|},
\quad \tilde{f}=\frac{\vec{U}_{ref}^3}{\rho_{0}}f.
\end{equation}
where the reference density and temperature are the unperturbed density and temperature, the reference velocity is the most probable speed $|\vec{U}_{ref}|=\sqrt{2k_BT_0/m}$. The distribution function $\tilde{f}(\vec{x},t,\tilde{\vec{v}})$ can be linearized with respect to the small perturbation $\delta$ \cite{sharipov2012rarefied},
\begin{equation}\label{f-linear}
  \tilde{f}(\vec{x},t,\tilde{\vec{v}})=\frac{1}{\pi^{3/2}} e^{-\tilde{\vec{v}}^2}[1+\tilde{f}_\delta(\vec{x},t,\tilde{\vec{v}})\delta],
\end{equation}
where the small perturbation $\delta$ can be a small pressure gradient, small temperature difference, small external force, etc.
The linearized moments can be obtained by
taking moments to above distribution function,
\begin{equation}\label{w-linear}
\begin{aligned}
  \tilde{\rho}(\vec{x},t)&=1+\tilde{\rho}_{\delta}(\vec{x},t)\delta,\\
  \tilde{\vec{U}}(\vec{x},t)&=\tilde{\vec{U}}_\delta(\vec{x},t)\delta,\\
  \tilde{T}(\vec{x},t)&=1+\tilde{T}_\delta(\vec{x},t)\delta,
\end{aligned}
\end{equation}
where
\begin{equation}
  \left(\begin{array}{c} \tilde{\rho}_{\delta}\\ \tilde{\vec{U}}_\delta\\\tilde{T}_\delta \\
  \end{array}\right)=
  \frac{1}{\pi^{3/2}}\int
  \left(\begin{array}{c} 1\\ \tilde{\vec{v}} \\
  \frac{2}{3}\tilde{\vec{v}}^2-1 \\ \end{array}\right)
  \text{e}^{-\tilde{\vec{v}}^2}
  \tilde{f}_\delta(\vec{x},t,\tilde{\vec{v}})d\tilde{\vec{v}}.
\end{equation}
The linearized BGK equation can be written as
\begin{equation}\label{bgk-linear}
  \frac{\partial \tilde{f}_\delta}{\partial t}+
  \tilde{\vec{u}}\cdot\frac{\partial \tilde{f}_\delta}{\partial \vec{x}}=
  \frac{1}{\tau}\left[\tilde{\rho}_\delta+
  2\tilde{\vec{u}}\cdot\tilde{\vec{U}}_\delta+
  \tilde{T}_\delta\left(\tilde{\vec{u}}^2-\frac{3}{2}\right)-\tilde{f}_\delta\right],
\end{equation}
where the linearized equilibrium distribution function $\tilde{g}_\delta $ is
\begin{equation}
\tilde{g}_\delta=\tilde{\rho}_\delta+
  2\tilde{\vec{u}}\cdot\tilde{\vec{U}}_\delta+
  \tilde{T}_\delta\left(\tilde{\vec{u}}^2-\frac{3}{2}\right).
\end{equation}
For the sake of simple notation, the tilde and subscript $\delta$ is omitted in the rest of this paper, and no confusion will be caused.

Analogy to the Chapman-Enskog theory, the linearized hydrodynamic equations can be derived from the above linearized kinetic equation in continuum regime \cite{chapman1990mathematical,guo2019unified}.
We perform asymptotic analysis to Eq.\eqref{bgk-linear} with respect to $\tau$, and the distribution function can be expanded as
\begin{equation}\label{expansion-f}
  f=f^{(0)}+\tau f^{(1)}+\tau^2f^{(2)}+...,
\end{equation}
and correspondingly the time derivative is also expanded as
\begin{equation}\label{expansion-t}
  \partial_t=\partial_{t_{0}}+\tau \partial_{t_{1}}+\tau^2 \partial_{t_{2}}+...,
\end{equation}
where $\partial_{t_{k}}$ stands for the contribution to $\partial_{t}$ from the spatial gradients of the $(k+1)-$th order hydrodynamic variables.
Substitute the expansion of the distribution function \eqref{expansion-f} and time derivative \eqref{expansion-t} into the linearized kinetic equation \eqref{bgk-linear}, and the following hierarchy and be obtained
\begin{align}
    &\epsilon^{-1}: \quad f^{(0)}=f^{eq},\label{hierarchy1}\\
    &\epsilon^0: \quad D_0f^{(0)}=-f^{(1)},\label{hierarchy2}\\
    &\epsilon^1: \quad \partial_{t_1}f^{(0)}+D_0f^{(1)}=-f^{(2)},\label{hierarchy3}\\
    &...\\
    &\epsilon^{k-1}: \quad \sum_{j=1}^{k-1}f^{(k-j-1)}+D_0f^{(k-1)}=f^{(k)},\label{hierarchy4}
\end{align}
where $D_0=\partial_{t_0}+\vec{v}\cdot\nabla$. The conservative constraint of collision operator gives
\begin{equation}
\int \psi f^{(k)} d\vec{v} =0,
\end{equation}
where $\psi=(1,\vec{v},\vec{v}^2/2)$ are the collision invariants.
Consider the two dimensional case, the second order hierarchy gives the linearized Euler equations
\begin{equation}\label{Euler}
  \frac{\partial}{\partial t}
  \left(\begin{array}{c}\rho\\U\\V\\\frac34(\rho+T)\\\end{array}\right)+
  \frac{\partial}{\partial x}
  \left(\begin{array}{c}U\\\frac12(\rho+T)\\0\\\frac54 U\\\end{array}\right)+
  \frac{\partial}{\partial y}
  \left(\begin{array}{c}V\\0\\\frac12(\rho+T)\\\frac54 V\\\end{array}\right)=0,
\end{equation}
where $U$, $V$ are the macroscopic x-directional velocity and y-directional velocity.
The second and third order hierarchies give the linearized Navier-Stokes equations
\begin{equation}\label{NS}
\begin{aligned}
  &\frac{\partial}{\partial t}
  \left(\begin{array}{c}\rho\\U\\V\\\frac34(\rho+T)\\\end{array}\right)+
  \frac{\partial}{\partial x}
  \left(\begin{array}{c}U\\\frac12(\rho+T)\\0\\\frac54 U\\\end{array}\right)+
  \frac{\partial}{\partial y}
  \left(\begin{array}{c}V\\0\\\frac12(\rho+T)\\\frac54 V\\\end{array}\right)\\
  &=
  \frac{\partial}{\partial x}
  \left(\begin{array}{c}0\\\frac23 \tau U_x-\frac13 \tau V_y\\\frac12\tau(V_x+U_y)\\\frac58 \tau T_x\\\end{array}\right)+
  \frac{\partial}{\partial y}
  \left(\begin{array}{c}0\\\frac12 \tau(U_y+V_x)\\\frac23\tau V_y-\frac13\tau U_x\\\frac58\tau T_y\\\end{array}\right),
\end{aligned}
\end{equation}
where the viscous coefficient $\mu=\frac{\tau}{2}$, heat conduction coefficient $\kappa=\frac{5\tau}{8}$, and the Prandtl number is $\frac{c_p\mu}{\kappa}=1$.
\subsection{Unified Gas-kinetic Scheme}
Consider two dimensional flow, the following reduced distribution functions are introduced to reduce the computational cost,
\begin{equation}
\begin{aligned}
  h(\vec{x},t,u,v)&=\frac{1}{\sqrt{\pi}}\int \text{e}^{-w^2}f(\vec{x},t,u,v,w) dw,\\
  b(\vec{x},t,u,v)&=\frac{1}{\sqrt{\pi}}\int \text{e}^{-w^2}\left(w^2-\frac12\right)f(\vec{x},t,u,v,w) dw.
\end{aligned}
\end{equation}
The moments of the reduced distribution function are
\begin{equation}
  \left(\begin{array}{c} \rho\\ \vec{U}\\ T \\
  \end{array}\right)=
  \frac{1}{\pi}\int
  \left(\begin{array}{c} h\\ \vec{v} h\\
  \frac23(u^2+v^2-1)h+\frac23b\\ \end{array}\right)
  \text{e}^{-(u^2+v^2)} dudv.
\end{equation}
The reduced distribution functions follow the kinetic equations
\begin{equation}\label{hb-linear}
\begin{aligned}
  \frac{\partial h}{\partial t}+
  u\frac{\partial h}{\partial x}+
  v\frac{\partial h}{\partial y}&=
  \frac{1}{\tau}\left[\rho+
  2\vec{v}\cdot\vec{U}+
  T(\vec{v}^2-1)-h\right],\\
  \frac{\partial b}{\partial t}+
  u\frac{\partial b}{\partial x}+
  v\frac{\partial b}{\partial y}&=
  \frac{1}{\tau}\left(\frac12T-b\right),
\end{aligned}
\end{equation}
where $g_h=\rho+2\vec{v}\cdot\vec{U}+T(\vec{v}^2-1)$ and $g_b=\frac12T$ are the reduced equilibrium distribution function.
The finite volume evolution equation of UGKS is obtained by integrating Eq.\eqref{hb-linear} with respect to space and time. Consider space control volume $\Omega_i$ and velocity control volume $\Omega_j$, and the cell averaged quantities are defined as
\begin{equation}
\begin{aligned}
  &h^{n}_{ij}=\frac{1}{|\Omega_{ij}|}\int h(\vec{x},t^n,\vec{v}) d\vec{x} d\vec{v},\\
  &b^{n}_{ij}=\frac{1}{|\Omega_{ij}|}\int b(\vec{x},t^n,\vec{v}) d\vec{x} d\vec{v},\\
  &\vec{W}^{n}_{ij}=\frac{1}{|\Omega_{i}|}\int \vec{W}(\vec{x},t^n) d\vec{x},
\end{aligned}
\end{equation}
where $\Omega_{ij}=\Omega_i\otimes\Omega_j$ is the control volume in the phase space.
The UGKS evolution equation of the distribution function is
\begin{equation}\label{update-f}
\begin{aligned}
h^{n+1}_{ij}=&h^{n}_{ij}-\frac{1}{|\Omega_i|}\int_{\partial \Omega_i} F^h_j ds
+\frac{\Delta t}{2}\left(\frac{g_{h,ij}^n-h^{n}_{ij}}{\tau^n}+\frac{g_{h,ij}^{n+1}-h^{n+1}_{ij}}{\tau^{n+1}}\right),\\
b^{n+1}_{ij}=&b^{n}_{ij}-\frac{1}{|\Omega_i|}\int_{\partial \Omega_i} F^b_j ds
+\frac{\Delta t}{2}\left(\frac{g_{b,ij}^n-b^{n}_{ij}}{\tau^n}+\frac{g_{h,ij}^{n+1}-b^{n+1}_{ij}}{\tau^{n+1}}\right),
\end{aligned}
\end{equation}
which is coupled with the evolution equation of the macroscopic conservative variables.
\begin{equation}\label{update-w}
  \vec{W}^{n+1}_i=\vec{W}^{n}_i-\frac{1}{|\Omega_i|}\int_{\partial \Omega_i} F^W_j ds.
\end{equation}
Assume that $t^n=0$, the center of cell interface $\vec{x}_{\partial\Omega_j}=0$, the projection of velocity on the outer normal vector $\vec{n}_{\partial\Omega_j}$ is $u$, and the UGKS multiscale numerical flux is
\begin{equation}\label{flux3}
\begin{aligned}
&F^h_j=\int_0^{\Delta t}ue^{-w^2} f(0,t,\vec{v}) dw dt,\\
&F^b_j=\int_0^{\Delta t}ue^{-w^2} (w^2-\frac12)f(0,t,\vec{v}) dw dt,\\
&F^W=\int_0^{\Delta t}u\vec{\psi} f(0,t,\vec{v}) d\vec{v} dt,
\end{aligned}
\end{equation}
where $\vec{\psi}=\left(1,\vec{v},\frac12\vec{v}^2\right)$ is the conservative moments.
The key ingredient of UGKS is the use of integral solution $f(0,t,\vec{v})$ in the construction of numerical flux,
\begin{equation}\label{integralsolution}
\begin{aligned}
  f(0,t,\vec{v})=&\frac{1}{\tau}\int_0^t g(\vec{x}',t',\vec{v}) e^{-\frac{t'-t}{\tau}} dt'+e^{-t/\tau} f_0(-\vec{v} t,\vec{v})\\
  =&(1-e^{-t/\tau})g_0
  +(\tau(e^{-t/\tau}-1)+te^{-t/\tau})(u\partial_xg_0+v\partial_yg_0)\\
  &+\tau(t/\tau-1+e^{-t/\tau})\partial_tg_0+e^{-t/\tau}f_0
  -te^{-t/\tau}(u\partial_xf_0 +v\partial_yf_0),
\end{aligned}
\end{equation}
where $f_0(\vec{x},\vec{v})$ is the initial distribution at $t=0$, and
\begin{equation}
\begin{aligned}
  &f_0=f_0^l H[u]+f_0^r (1-H[u]),\\
  & \int \psi g_0 d\vec{v} = \int \psi f_0 d\vec{v}, \\
  &\partial_x g_0=\partial_x^l g_0 H[u]+\partial_x^r g_0 (1-H[u]),\\
  &\partial_x f_0=\partial_x^l f_0 H[u]+\partial_x^r f_0 (1-H[u]),
\end{aligned}
\end{equation}
where $H[x]$ is the Heaviside function, and the least squares method is used for spatial reconstruction.
The time derivative is approximated by the first order Chapman-Enskog expansion,
\begin{equation}
W_t=-\int (u\partial_x g+v\partial_y g) d\vec{v}.
\end{equation}
Substitute the integral solution Eq.\eqref{integralsolution} into the numerical flux Eq.\eqref{flux3}, and we have
\begin{equation}\label{fluxf}
\begin{aligned}
  F^h_j=&c_1u_j(\rho_0 +2\vec{v}\cdot\vec{U}_0+T_0(\vec{v}^2-1))\\
  &+c_2u_j^2(\rho_x^l +2\vec{v}\cdot\vec{U}_x^l+T_x^l(\vec{v}^2-1))H[u_j]\\
  &+c_2u_j^2(\rho_x^r +2\vec{v}\cdot\vec{U}_x^r+T_x^r(\vec{v}^2-1))(1-H[u_j])\\
  &+c_2u_jv_j(\rho_y +2\vec{v}\cdot\vec{U}_y+T_y(\vec{v}^2-1))\\
  &+c_3u_j(\rho_t +2\vec{v}\cdot\vec{U}_t+T_t(\vec{v}^2-1))\\
  &+c_4u_jh_0+c5(u_j^2h_{x,j}+u_jv_jh_{y,j}),\\
  F^b_j=&c_1u_jT_0/2+c_2u_j^2T_x^l/2H[u_j]+c_2u_j^2T_x^r(1-H[u_j])\\
  &+c_2v_ju_jT_y+c_3u_jT_t+c_4u_jb_0+c5(u_j^2b_{x,j}+u_jv_jb_{y,j}),
\end{aligned}
\end{equation}
and
\begin{equation}\label{fluxw}
\begin{aligned}
F^W=&c_1 \int u\vec{\psi} g d\vec{v}
+c_2 \int u^2\vec{\psi} \big[g_x^l H[u]+g_x^r (1-H[u])\big] d\vec{v}
+c_2 \int uv\vec{\psi} g_y d\vec{v}
+c_3 \int u\vec{\psi} g_t d\vec{v}\\
&+c_4\left(
       \begin{array}{c}
         \sum \omega_k u_k h_k \\
         \sum \omega_k u_k^2 h_k \\
         \sum \omega_k u_kv_k h_k \\
        \frac12 \sum \omega_k (u_k^3 h_k+u_k b_k) \\
       \end{array}
     \right)
     +c_5\left(
       \begin{array}{c}
         \sum \omega_k (u_k^2 h_{x,k}+u_kv_kh_{y,k}) \\
         \sum \omega_k (u_k^3 h_{x,k}+u_k^2v_k h_{y,k}) \\
         \sum \omega_k (u_k^2v_k h_{x,k}+u_kv_k^2 h_{y,k}  )\\
        \frac12 \sum \omega_k (u_k^4 h_{x,k}+u_k^2 b_{x,k}+u_k^3v_k h_{y,k}+u_kv_k b_{y,k}) \\
       \end{array}
     \right)
\end{aligned}
\end{equation}
where the time integration coefficients are
 \begin{equation}\label{coefficients}
   \begin{aligned}
     &c_1=1-\frac{\tau}{\Delta t}\left(1-e^{-\Delta t/\tau}\right),\\
     &c_2=-\tau+\frac{2\tau^2}{\Delta t}-e^{-\Delta t/\tau}\left(\frac{2\tau^2}{\Delta t}+\tau\right),\\
     &c_3=\frac12\Delta t-\tau+\frac{\tau^2}{\Delta t}\left(1-e^{-\Delta t/\tau}\right)\\
     &c_4=\frac{\tau}{\Delta t}\left(1-e^{-\Delta t/\tau}\right),\\
     &c_5=\tau e^{-\Delta t/\tau}-\frac{\tau^2}{\Delta t}\left(1-e^{-\Delta t/\tau}\right).
   \end{aligned}
 \end{equation}
The evolution equations Eq.\eqref{update-f},\eqref{update-w} and multiscale flux \eqref{fluxf},\eqref{fluxw} close the UGKS formulation. In the next section we are going to analyse the unified preserving property of UGKS.

\section{Unified Preserving Property of UGKS}
The concept of unified preserving property was proposed by Guo et al. to measure capability of the schemes in capturing the multiscale flow physics \cite{guo2019unified}. If a scheme is able to preserve the $n$-th order Chapman-Enskog expansion in continuum regime, with $\Delta t<\text{Kn}^\alpha$, then the scheme is an $n$-th order UP scheme. In this section, we are going to prove that the UGKS is a second order UP scheme.

\begin{theorem}\label{uptheorem}
  For $\Delta t< O(\epsilon^{1/2})$ and $\Delta x< O(\epsilon^{1/2})$, the Chapman-Enskog expansion coefficients of $f$ obtained from UGKS satisfy
  \begin{equation}
  f^{(0)}=f^{(eq)}, \quad \sum_{j=1}^{k-1}f^{(k-j-1)}+D_0f^{(k-1)}=f^{(k)} (1\leq k\le n),
  \end{equation}
  for $n=2$, if the relaxation time $\tau$ is a constant.
\end{theorem}
\begin{proof}
  Without loss of generality, consider the one dimensional case.
  The numerical flux of UGKS at cell interface $x_{i-1/2}$ can be written as
  \begin{equation}
  \begin{aligned}
  F_{i-\frac12}=&\frac{1}{\Delta t}\int_{t^n}^{t^{n+1}} vf_{i-\frac12}(t,\vec{v})dt\\
  =&c_1vg_{i-\frac12}+c_2v^2\partial_xg_{i-\frac12}+c_3v\partial_tg_{i-\frac12}+c_4vf_{i-\frac12}+c_5v^2\partial_xf_{i-\frac12}.
  \end{aligned}
  \end{equation}
  For $0\leq\Delta t\leq \tau$, namely $\Delta t=\tau^\alpha$ and $\alpha>1$. The coefficients $c_{1-5}$ can be expanded up to $\left(\frac{\Delta t}{\tau}\right)^2=O(\tau^{2\alpha-2})$ as
  \begin{equation}
  \begin{aligned}
  c_1&=\frac12 \frac{\Delta t}{\tau}+O\left(\tau^{2\alpha-2}\right),\\
  c_2&=0+O\left(\tau^{2\alpha-1}\right),\\
  c_3&=0+O\left(\tau^{2\alpha-1}\right),\\
  c_4&=\Delta t-\frac{\Delta t}{2\tau}+O\left(\tau^{2\alpha-2}\right),\\
  c_5&=\frac12 \Delta t^2+O\left(\tau^{2\alpha-1}\right),
  \end{aligned}
  \end{equation}
  and therefore, we have
  \begin{equation}
  \frac{F_{i+\frac12}-F_{i-\frac12}}{\Delta x}=v\partial_xf+\frac{\Delta t}{2}\big[-v^2\partial_xf+v\partial_xQ\big]
  +O\left(\tau^{2\alpha-1}\right)\mathcal{L}_{F1}(Q)+O(\tau^{2\alpha})\mathcal{L}_{F2}(f,g),
  \end{equation}
  where $\mathcal{L}$ notates for a generalized linearized operator, and $Q=(g-f)/\tau$ is the collision operator.
  We can obtain the modified equation of UGKS as
  \begin{equation}
  \partial_t f+v\partial_xf-Q=\frac{\Delta t}{2}
  \underbrace{\big[\partial_t^2f-v^2\partial_x^2f
  -\partial_tQ+u\partial_xQ\big]}_{\mathcal{A}}
  +O\left(\tau^{2\alpha-1}\right)\mathcal{L}_{F1}(Q)+O(\tau^{2\alpha})\mathcal{L}_{F2}(f,g),
  \end{equation}
  The underbraced term $\mathcal{A}$ can be estimated as
  \begin{equation}
\begin{aligned}
  \mathcal{A}=&(\partial_t-v\partial_x)\big[\partial_tf+v\partial_xf-Q\big]\\
  =&-\frac{\Delta t}{2}(\partial_t-v\partial_x)\mathcal{A}+O\left(\tau^{2\alpha-1}\right)\mathcal{L}_{F1}(Q)+O(\tau^{2\alpha})\mathcal{L}_{F2}(f,g),
\end{aligned}
  \end{equation}
  which gives
  \begin{equation}
  A=O\left(\tau^{2\alpha-1}\right)\mathcal{L}_{F1}(Q)+O(\tau^{2\alpha})\mathcal{L}_{F2}(f,g),
  \end{equation}
  and UGKS modified equation can be written as
   \begin{equation}\label{modified1}
  \partial_t f+v\partial_xf-Q=
  O\left(\tau^{2\alpha-1}\right)\mathcal{L}_{1}(Q)+O(\tau^{2\alpha})\mathcal{L}_{2}(f,g).
  \end{equation}
  As shown in figure \ref{order1}, only $O(\tau^{-1})$ in $\mathcal{L}_1$ is included in the first three orders of expansion. And the Chapman-Enskog hierarchy can be obtained as following
\begin{equation}
  \begin{aligned}
    &\epsilon^{-1}: \quad f^{(0)}=f^{eq},\\
    &\epsilon^0: \quad D_0f^{(0)}=-f^{(1)},\\
    &\epsilon^1: \quad \partial_{t_1}f^{(0)}+D_0f^{(1)}=-f^{(2)}.
\end{aligned}
\end{equation}
  \begin{figure}
   \centering
   \includegraphics[width=0.8\textwidth]{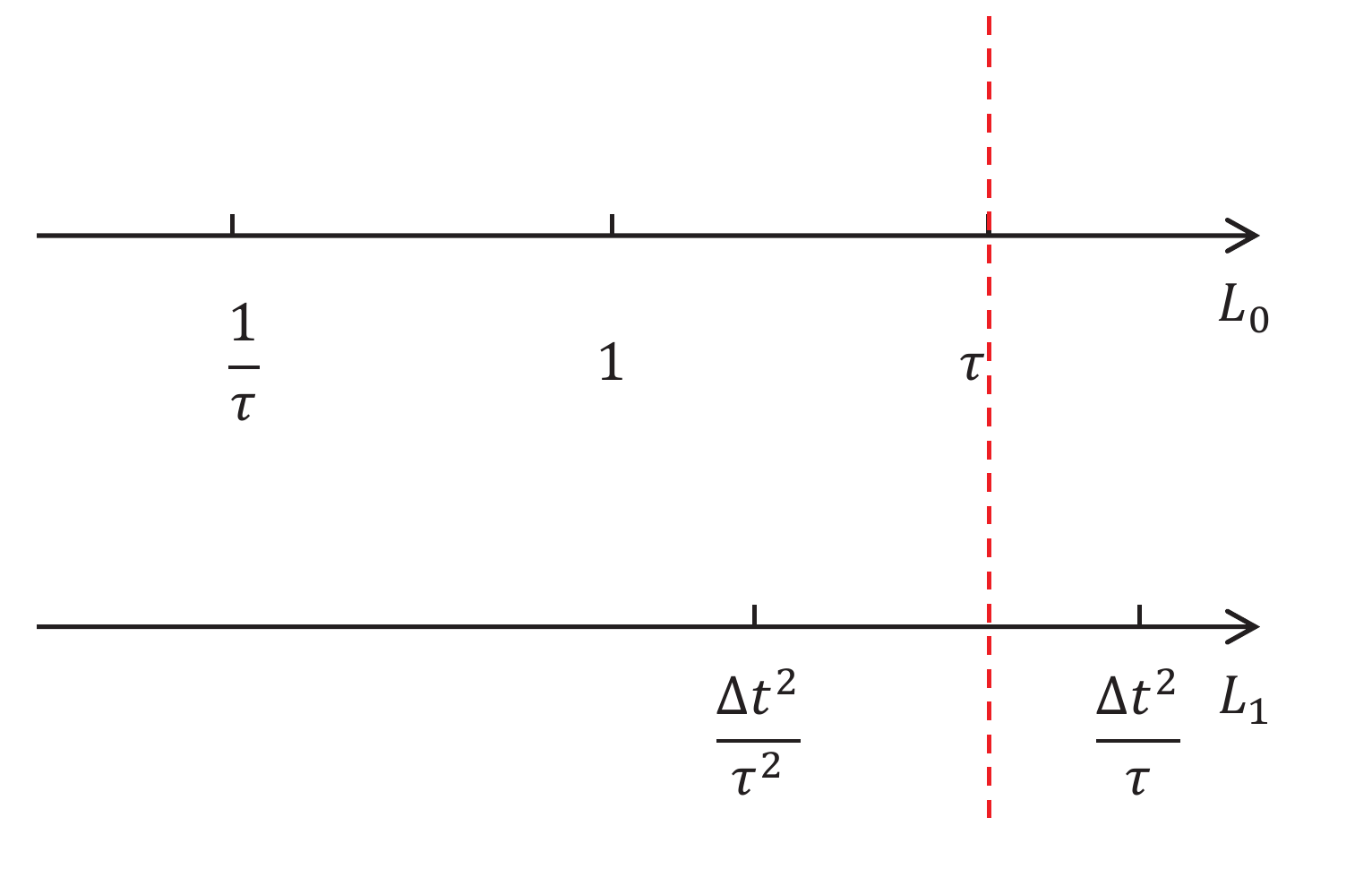}
   \caption{Order sequence of the modified equation \eqref{modified1}.}
   \label{order1}
   \end{figure}
   For $\tau<\Delta t<\tau^{0.5}$, namely $\Delta t=\tau^\beta$ and $0.5<\beta<1$. The coefficients $c_{1-5}$ can be estimated as
   \begin{equation}
  \begin{aligned}
  c_1&=\Delta t-\tau,\\
  c_2&=2\tau^2-\tau\Delta t,\\
  c_3&=\frac12\Delta t^2-\tau\Delta t+\tau^2,\\
  c_4&=\tau,\\
  c_5&=\tau^2,
  \end{aligned}
  \end{equation}
  and we can obtain the modified equation of UGKS as
  \begin{equation}\label{modified2}
\begin{aligned}
  &\partial_tf+v\partial_x\left[g-\tau(g_t+vg_x)\right]-Q\\
  =&\frac{\Delta t}{2}\mathcal{L}_1+\frac{\tau^2}{\Delta t}
  \mathcal{L}_2+O(\Delta t^2)\mathcal{L}_3(Q)+O(\Delta x^3)\mathcal{L}_4(g,f),
\end{aligned}
\end{equation}
where
\begin{equation}
\begin{aligned}
  &\mathcal{L}_1=\partial_t^2f
  +v\partial_t\partial_xg
  +\partial_tQ\\
  &\mathcal{L}_2=2v^2\partial_x^2g+v\partial_t\partial_xg
  -u^2\partial_x^2f+v \partial_xQ
\end{aligned}
\end{equation}

\begin{figure}
   \centering
   \includegraphics[width=0.8\textwidth]{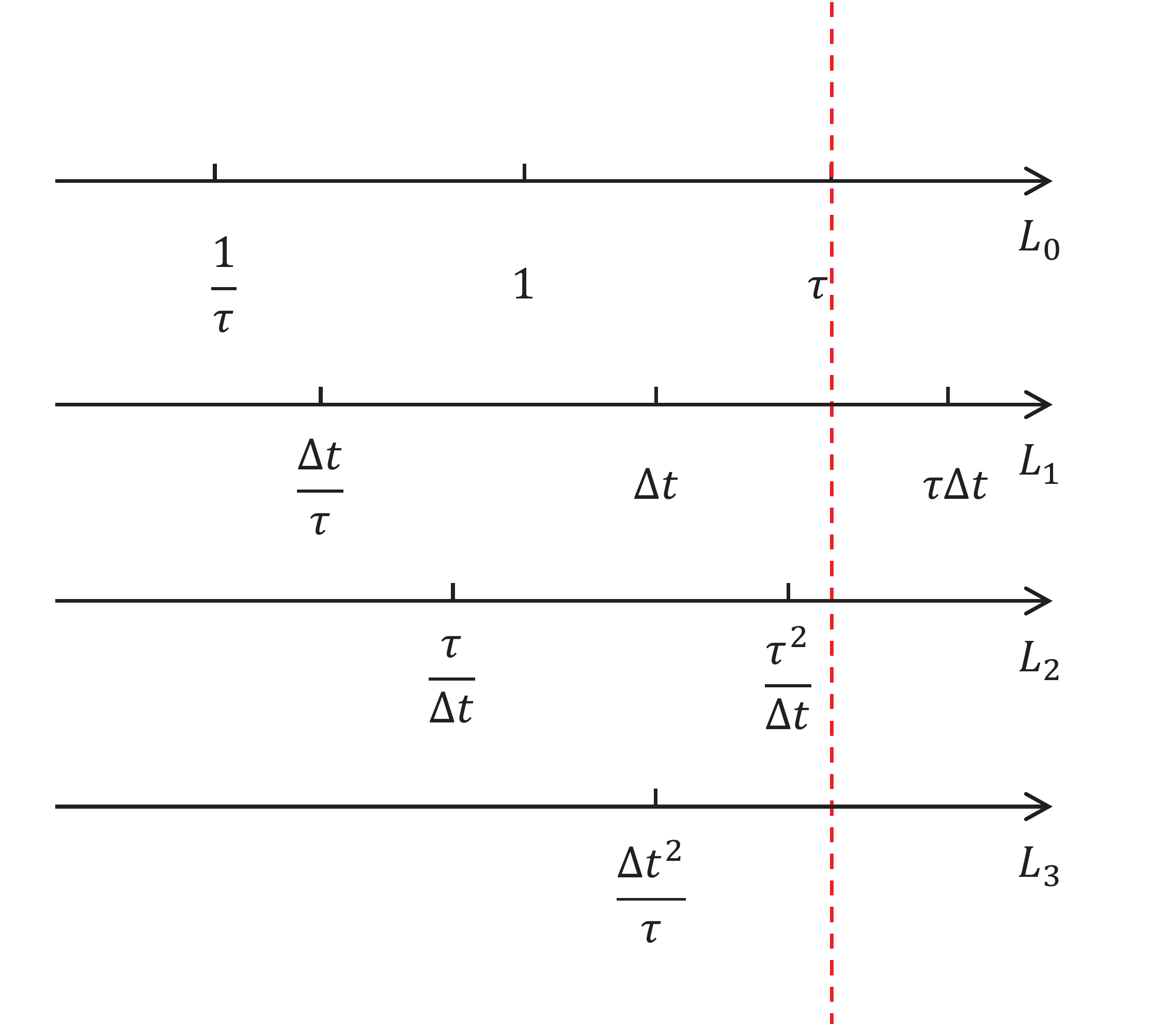}
   \caption{Order sequence of the modified equation \eqref{modified2}.}
   \label{order2}
   \end{figure}
  As shown in Fig.\ref{order2}, only $O(\tau^{-1})$ in $\mathcal{L}_1$ is included in the first three orders of expansion. And the Chapman-Enskog hierarchy can be obtained as following
\begin{equation}
  \begin{aligned}
    &\epsilon^{-1}: \quad f^{(0)}=f^{eq},\\
    &\epsilon^0: \quad D_0f^{(0)}=-f^{(1)},\\
    &\epsilon^1: \quad \partial_{t_1}f^{(0)}+D_0f^{(1)}=-f^{(2)}.
\end{aligned}
\end{equation}
It is shown that for $\Delta t< O(\epsilon^{1/2})$ and $\Delta x< O(\epsilon^{1/2})$, the UGKS exactly preserves the second order Chapman-Enskog expansion. Therefore the UGKS is a second order unified preserving scheme.
\end{proof}

Note that theorm \ref{uptheorem} holds for both linear and nonlinear UGKS. The conventional way to analyze the multiscale property of UGKS is to perform discrete asymptotic analysis to the discrete governing equations of UGKS, and then compare the discrete asymptotic analysis to the continuous Chapman-Enskog solution, which has been done by Liu et al. \cite{liu2016}. As shown in Fig. \ref{path}, the conventional analysis and current UP analysis are equivalent, and both show that the UGKS preserves the NS solution in continuum regime even when the cell size and time step are much larger than the kinetic particle mean free path and collision time.
\begin{figure}
\centering
\includegraphics[width=0.8\textwidth]{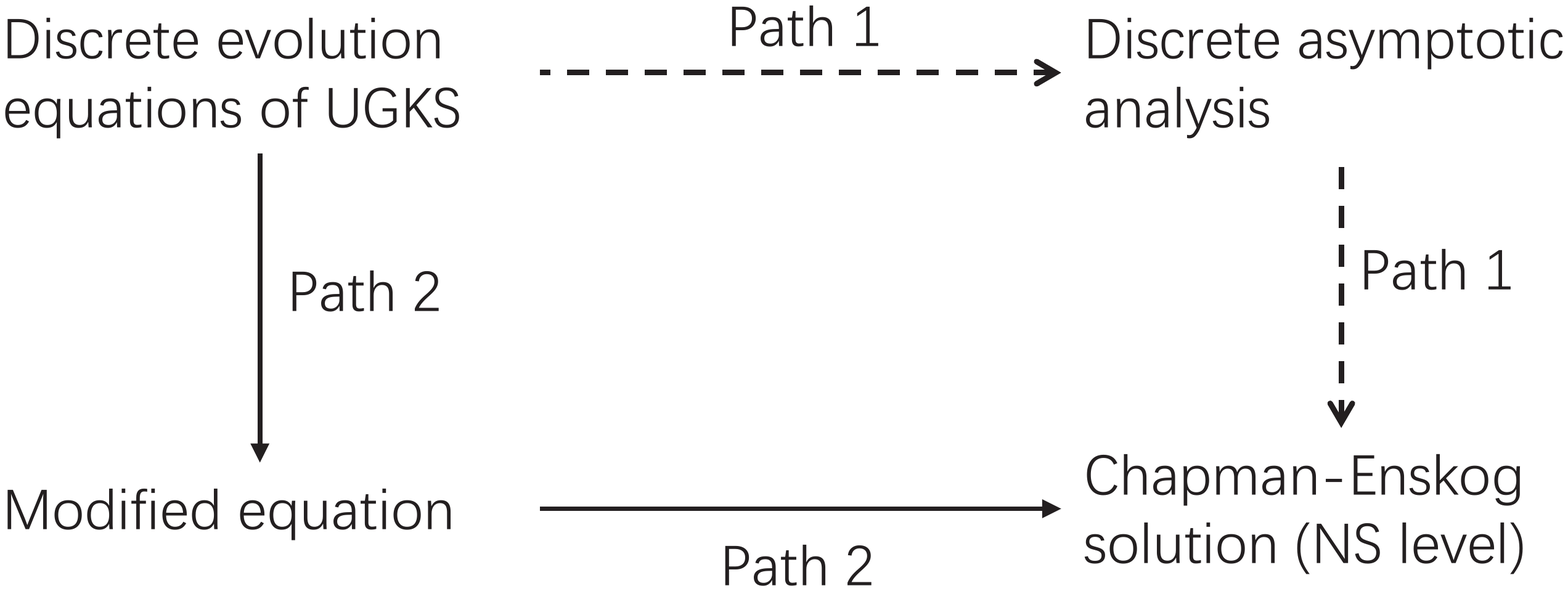}
\caption{Two equivalent ways to analyze the multiscale property of UGKS.}
\label{path}
\end{figure}

\section{Numerical Tests}
We perform five numerical tests to verify the accuracy and multiscale property of UGKS, including a 1D poiseuille flow test and four 2D tests. The Knudsen number of the numerical tests varies from 10 to $10^{-4}$, covering the flow regime from highly rarefied to Navier-Stokes regimes. It can be observed from the comparison that the UGKS well captures the kinetic solution in rarefied regime, and is able to capture NS solution with cell size much larger than the particle mean free path.

\subsection{Poiseuille flow}
The first test case is the one dimensional poiseuille flow.
The argon gas is confined between two isothermal wall located at $x=0$ and $x=1$.
The y-directional external force $\vec{F}=F_\delta \vec{y}$ is small, and all flow quantities are expanded with respect to the small force $F_\delta$. Two cases with different Knudsen numbers are calculated. For the first test case, the Knudsen number is 1.0, the physical space is divided into 40 equally distributed cells and the numerical space $[-5,5]$ is divided into 32 equally distributed velocity points. The solution of UGKS is compared to the kinetic solution. The second case is in the continuum regime with Knudsen number $10^{-4}$, the physical space is divided into 40 equally distributed cells and the 8 Gauss-Hermite quadrature is used in velocity space. The UGKS solution is compared to the analytic NS solution. As shown in Fig. \ref{poi}, the UGKS solution well agrees with the reference solution. For the Knudsen number $10^{-4}$ case, the time step is about 200 times of the relaxation time, and in such a case, the traditional upwind flux based DVM solution significantly deviates from the analytical one due to large numerical dissipation \cite{chen2015comparative}. However, the integral solution based multiscale flux of UGKS accurately recovers the NS flux in such a flow regime.
\subsection{Micro flow through periodic square cylinders}
The second test case is about a pressure gradient driven flow passing through an array of square cylinders. One replicated square is picked as our computational domain. Periodic boundary condition is used for the computational domain boundary, and the solid boundary with accommodation $\alpha=1$ is used for the cylinder boundary. Two types of solid squares are considered, namely a solid square and a caved square.
For the solid square case, three Knudsen numbers $\text{Kn}=10^{-1},10^{-2},10^{-4}$ are calculated.
The size of spatial cell is $\Delta x=1/120$. For $\text{Kn}=10^{-1},10^{-2}$ the velocity space $[-5,5]\times[-5,5]$ are equally divided into $32\times32$ points, and for $\text{Kn}=10^{-4}$, the $8\times8$ Gauss-Hermite quadrature is used for velocity space. The UGKS results are shown in Figs. \ref{square1} to \ref{square3}.
For the $\text{Kn}=10^{-4}$ case, the numerical cell size is 83 times of the particle mean free path. We compare the streamline and the velocity profile along $y=0.25$ of UGKS and NS solution as shown in Figs. \ref{square2} and \ref{square3}. It can be observed that the NS solutions are well captured even with a cell size much larger than the kinetic scale.
Similar to the solid square, for the caved square case, three Knudsen numbers $\text{Kn}=10^{-1},10^{-2},10^{-4}$ are calculated and the results are shown in Figs. \ref{square4} to \ref{square6}. In continuum regime, the UGKS well agrees with NS solution as shown in Figs. \ref{square5} to \ref{square6}.

\subsection{Thermal creep micro flow}
For the third test case, we study the micro flow driven by a small temperature gradient. Consider a $1.0\times0.25$ rectangular cavity. The temperature of the left and right wall is $T^l=0$ and $T^r=1.0$. The temperature distribution of the top and bottom wall is $T^{t,b}=x$. We consider two flow regime with $\text{Kn}=10$ and $\text{Kn}=10^{-2}$.
The spatial cell size is set as $\Delta x=1/120$, and the velocity space $[-5,5]\times[-5,5]$ are equally divided into $32\times32$ points.
Similar to the nonlinear case \cite{huang2013}, a counterclockwise and clockwise streamline is formed on the top and bottom region of the cavity for the rarefied case, while a reversed streamline is formed in the continuum regime, as shown in Fig. \ref{creep1}. We also compare the UGKS solution to the NS solution in Fig. \ref{creep2}. It can be observed that the UGKS solution agrees well with NS solution for density, velocity and temperature distribution.

\subsection{Flow induced by a hot microbeam}
We study the flow induced by a hot microbeam in the transitional flow regime by UGKS and compare with the solutions with the R26 moment method \cite{sheng2014simulation} and the GSIS \cite{su2020can}. The numerical setup is the same as Zhu et al. \cite{zhu2019application}. A cavity with isothermal boundary is located at $[0,10]\times[0,8]$, inside which a hot microbeam is located at $[1,5]\times[1,3]$. The temperature of the cavity boundary is $T=0$ and the temperature of the microbeam boundary is $T=1$. The transitional flow regime with Knudsen number $\text{Kn}=5\times10^{-3}$ is simulated, and the $8\times8$ Gauss-Hermite quadrature is used for velocity space. After reaching the steady state, two thermal gradient induced vertexes will be formed at each corner of the microbeam. Three meshes have been used for UGKS simulation: a uniform mesh with $\Delta x=4\lambda$, a uniform mesh with $\Delta x=\lambda$, and a nonuniform mesh with the minimum cell size $\Delta x=0.1\lambda$, where $\lambda$ is the particle mean free path.
For the first two sets of mesh, the linearized UGKS is used, and for the third set of mesh, the nonlinear UGKS is used \cite{zhu2017implicit}.
The velocity magnitude and temperature distribution for $\Delta x=4\lambda$ mesh is shown in Fig. \ref{microbeam1}, and the streamline and heat flux is shown in Fig. \ref{microbeam2}. The comparison of the x-velocity profile along $x=0.5$, and y-velocity profile along $y=0.5$ is shown in Fig. \ref{microbeam3}. It can be observed that the results of UGKS under $\Delta x=4 \lambda$ agrees with the converged solution of R26 and GSIS, while the velocity magnitude of UGKS decreases as mesh gets refined. Limited by the large computational cost, the finest mesh we use for UGKS is $\Delta x=0.1\lambda$, which is quite close to the UGKS converged solution. It can be observed that the magnitude of UGKS converged solution is half the R26 and GSIS results, and further verification is needed.

\subsection{Lid-driven cavity flow}
The last test case is the simulation of the lid-driven cavity flow. The cavity is located at $[0,1]\times[0,1]$, and the top lid is moving towards positive x direction with a small velocity $\vec{U}=U_{\delta}\vec{x}$. Two Knudsen numbers are considered, namely $\text{Kn}=10^{-1}$ and $\text{Kn}=10^{-4}$. The spatial cell size for both cases are $\Delta x=0.01$, and the velocity space for $\text{Kn}=10^{-1}$ is $[-5,5]\times[-5,5]$ divided by $32\times32$ velocity points, and $8\times8$ Gauss-Hermite quadrature is used for $\text{Kn}=10^{-4}$ case. The results of UGKS is compared to NS solution as shown in Figs. \ref{cavity1}-\ref{cavity4}. For the rarefied case, it can be observed that the NS equations break down, especially for the heat flux calculation. The special heat transfer from hot region to cold region is not captured by NS solution, while the velocity profile doesn't deviate that far. In the continuum regime, the UGKS recovers the NS solution. Especially for the velocity field, the UGKS and NS solutions are identical, even with the UGKS cell size 100 times larger than the particle mean free path.
For this test case, the linearized UGKS is about 2.5 times faster than the nonlinear UGKS under same cell size and velocity points, while the flux calculation of linearized UGKS is about 3.5 times faster than nonlinear UGKS.

\section{Conclusion}
In this paper, we extend the UGKS to the micro flow simulation. Compare to the nonlinear UGKS, the linearized UGKS is faster and more accurate for the micro flow simulation. The multiscale property of UGKS is inherited by the linearized UGKS. In the rarefied regime, the linearized UGKS well captures the linear kinetic solution. In the continuum regime, the viscous solution can be accurately captured by UGKS even with cell size and time step much larger than the particle mean free path and collision time.
Theoretically, we prove the unified preserving property of UGKS, which shows that UGKS is a second order UP scheme.
The proof holds for both the linearized UGKS and nonlinear UGKS.
In term of flux calculation, the linearized UGKS is more than three times faster than the nonlinear UGKS. Combining the linearized UGKS and implicit technique such as LU-SGS and multigrid \cite{zhu2016implicit,zhu2017implicit}, the UGKS will be a powerful numerical tool for the study of micro flow in the fields of porous media and MEMS.

\section*{Funding}
The current research is supported by Hong Kong research grant council (16206617)
and National Science Foundation of China (11772281, 91852114), and the National Numerical Windtunnel project. 

\section*{Availability of data and materials}

All data and materials are available upon request.

\subsection*{Availability of supporting data}
 Not applicable

\section*{Authors’ contributions}
Our group has been working on the topic for a long time.
The research output is coming from our joint effort.
All authors read and approved the final manuscript.

\section*{Competing interests}
The authors declare that they have no competing interests.

\begin{figure}
\centering
\includegraphics[width=0.48\textwidth]{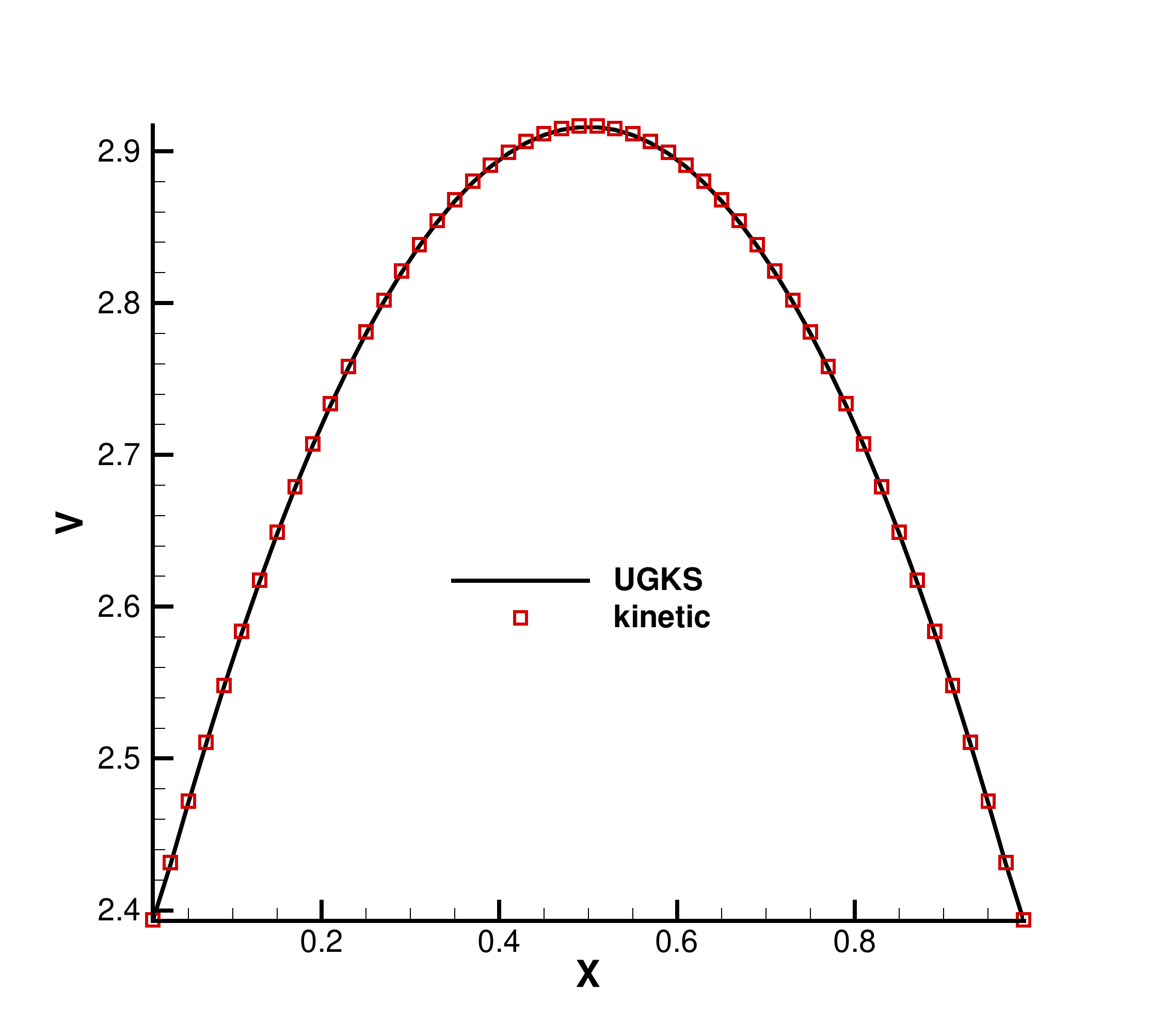}
\includegraphics[width=0.48\textwidth]{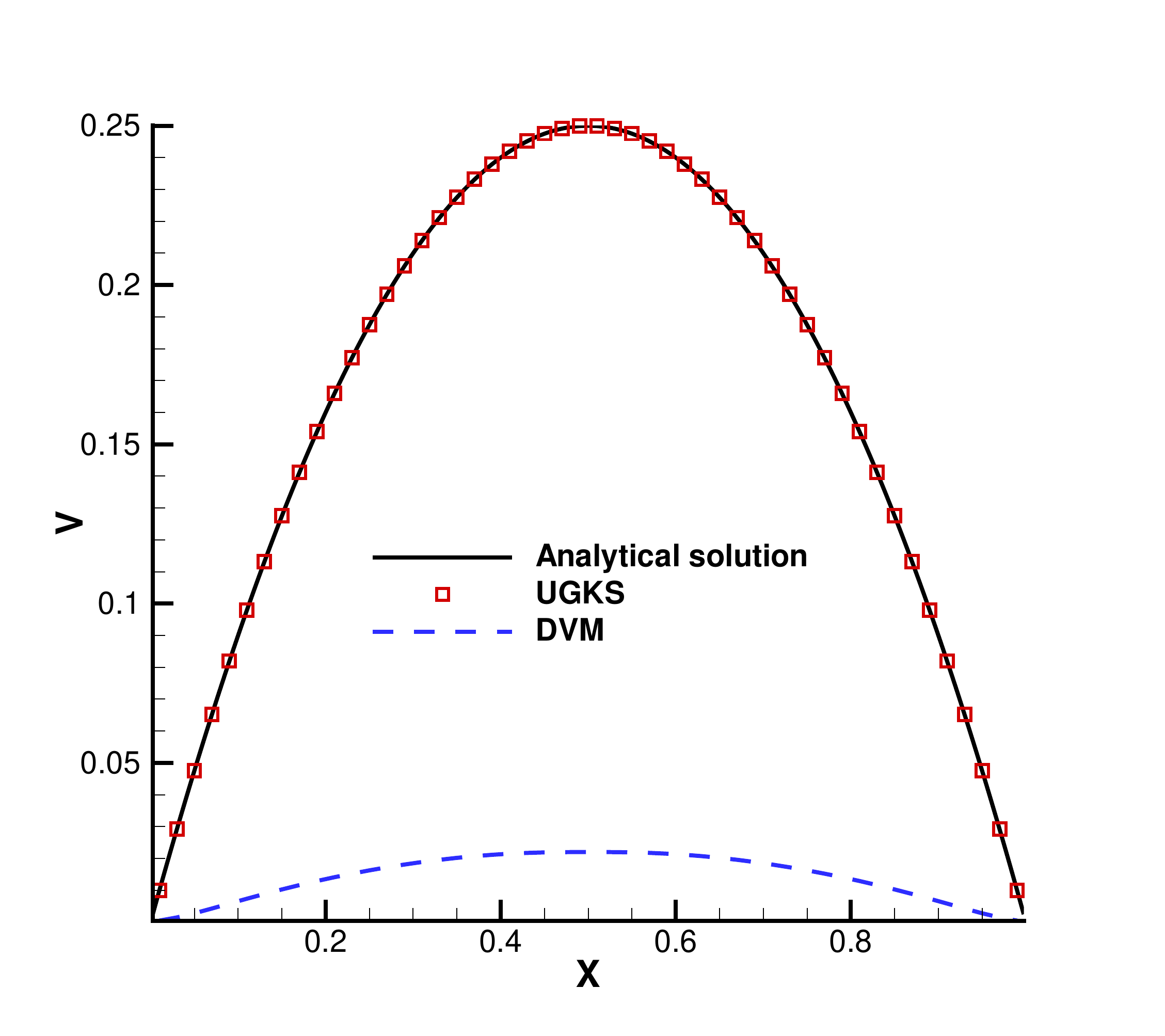}
\caption{The steady solution of Poiseuille flow. Symbol shows the UGKS solution, comparing to the analytical solution and traditional discrete ordinate method solution. Left figure shows the solution with Knudsen number $\text{Kn}=1.0$, and right figure shows the solution with Knudsen number $\text{Kn}=1.0\times10^{-4}$.}
\label{poi}
\end{figure}

\begin{figure}
\centering
\includegraphics[width=0.48\textwidth]{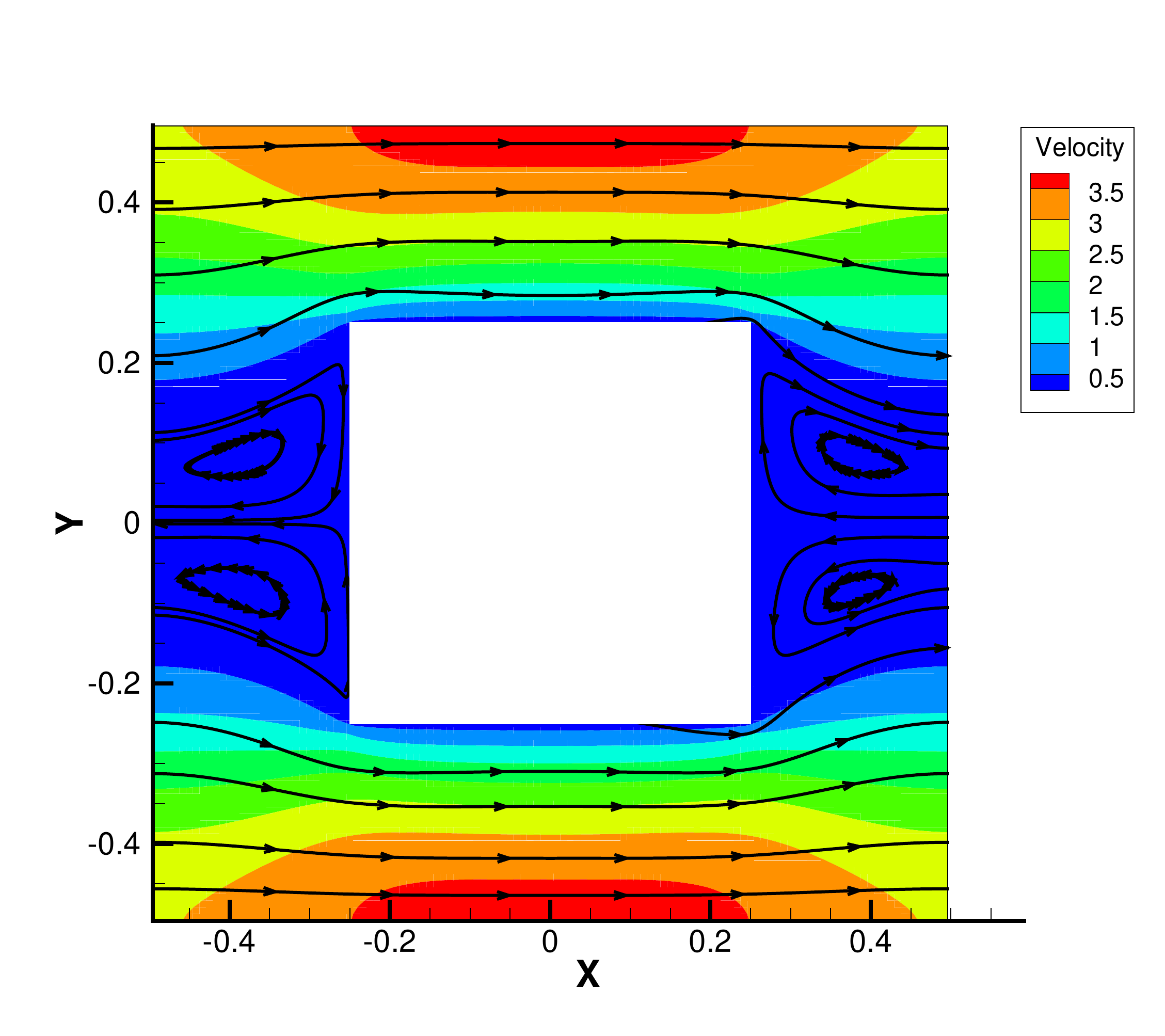}
\includegraphics[width=0.48\textwidth]{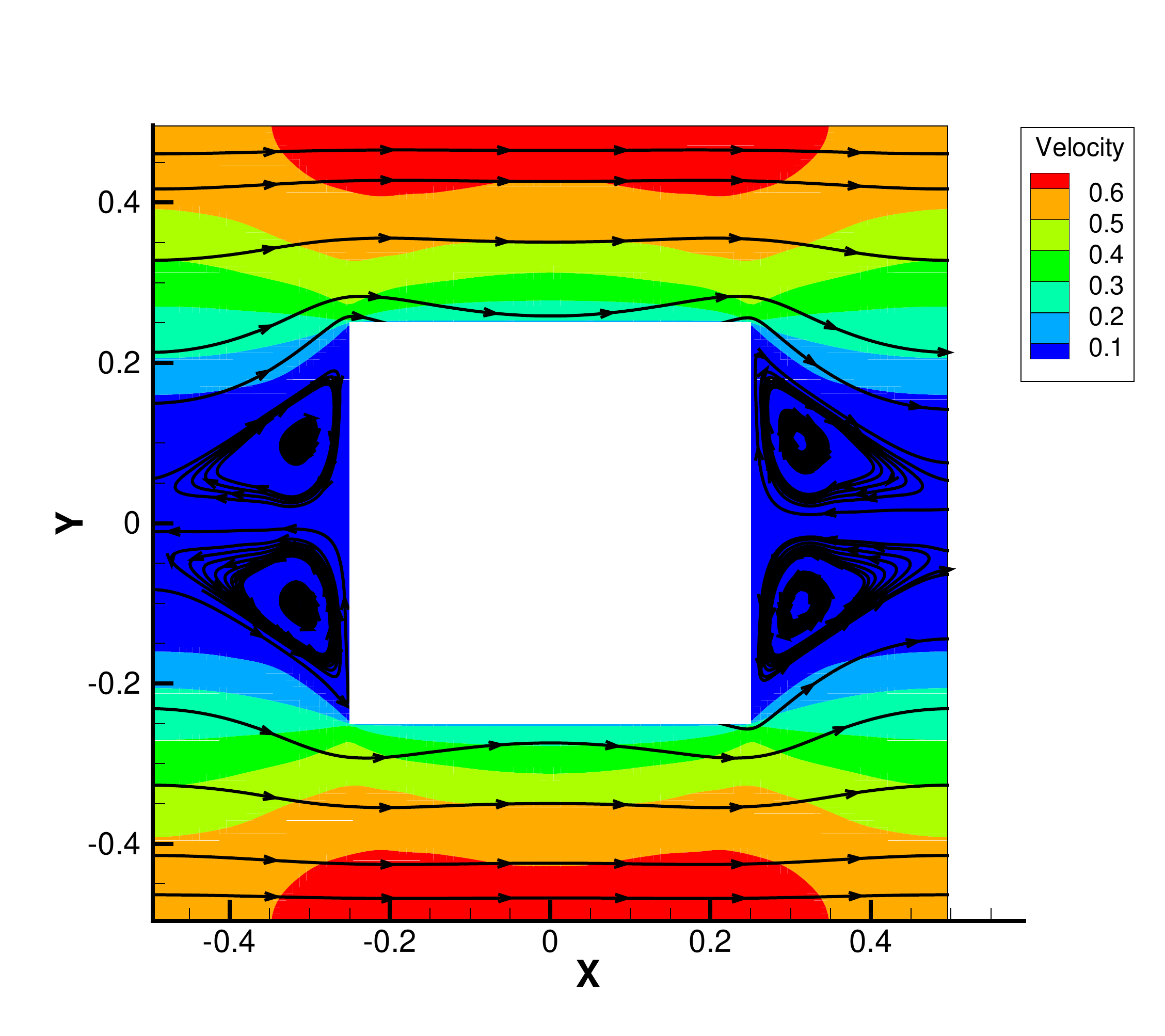}
\caption{The streamline and velocity magnitude of the micro flow through periodic square cylinders. Left figure shows the solution with Knudsen number $\text{Kn}=1.0\times10^{-1}$, and right figure shows the solution with Knudsen number $\text{Kn}=1.0\times10^{-2}$.}
\label{square1}
\end{figure}

\begin{figure}
\centering
\includegraphics[width=0.48\textwidth]{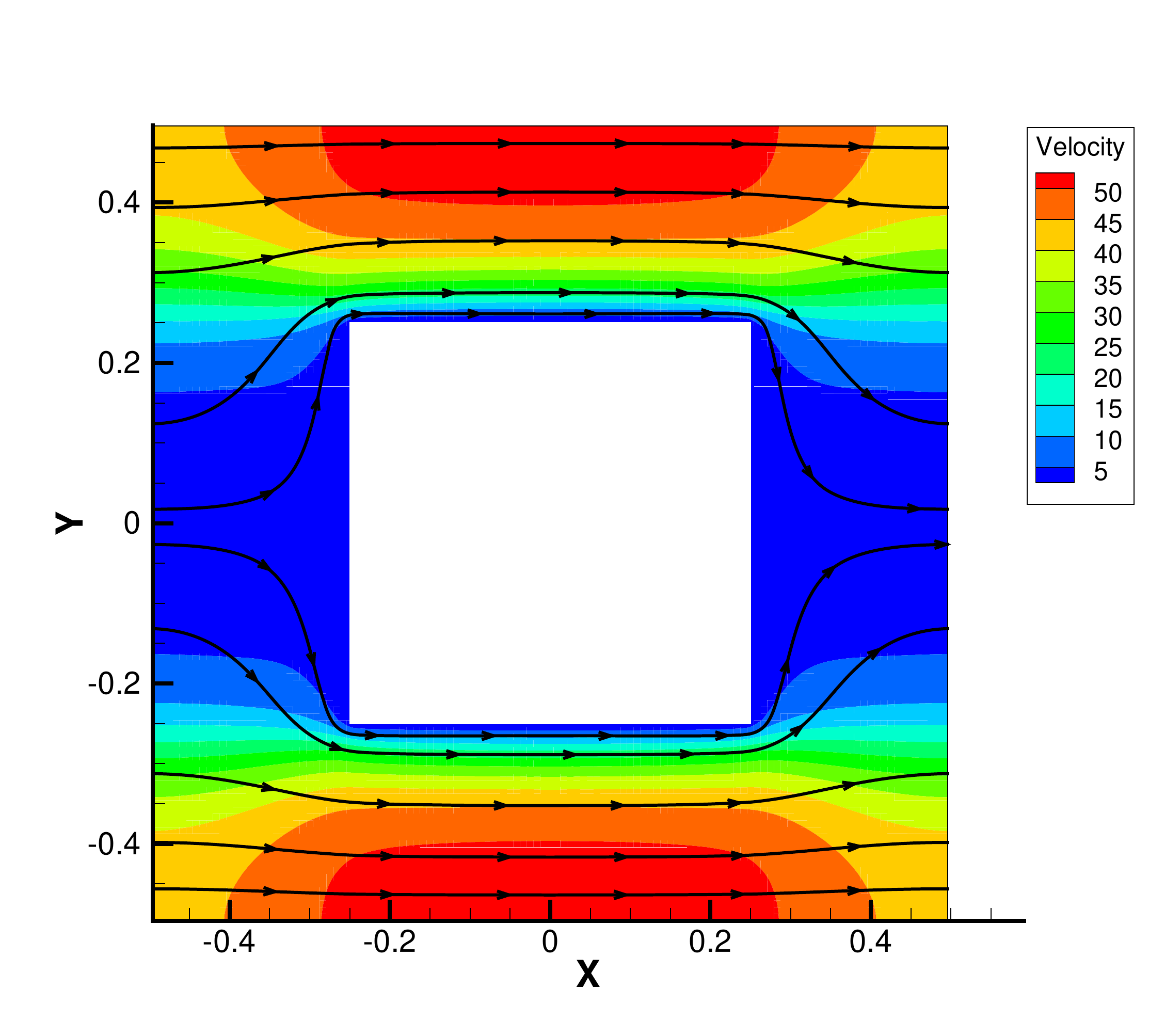}
\includegraphics[width=0.48\textwidth]{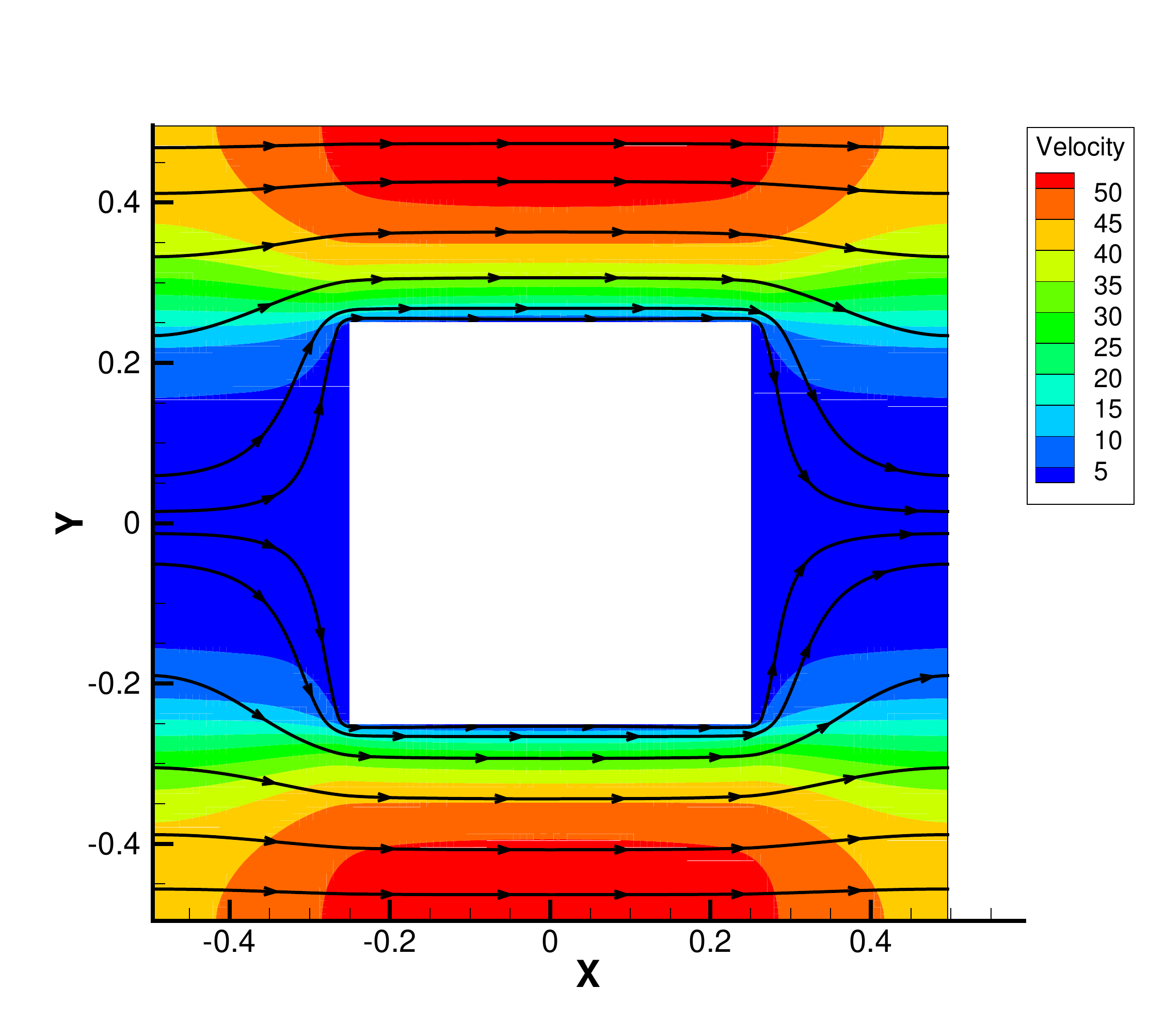}
\caption{The streamline and velocity magnitude of the micro flow through periodic square cylinders with Knudsen number $\text{Kn}=1.0\times10^{-4}$. Left figure shows UGKS solution and right figure shows the NS solution by GKS.}
\label{square2}
\end{figure}

\begin{figure}
\centering
\includegraphics[width=0.48\textwidth]{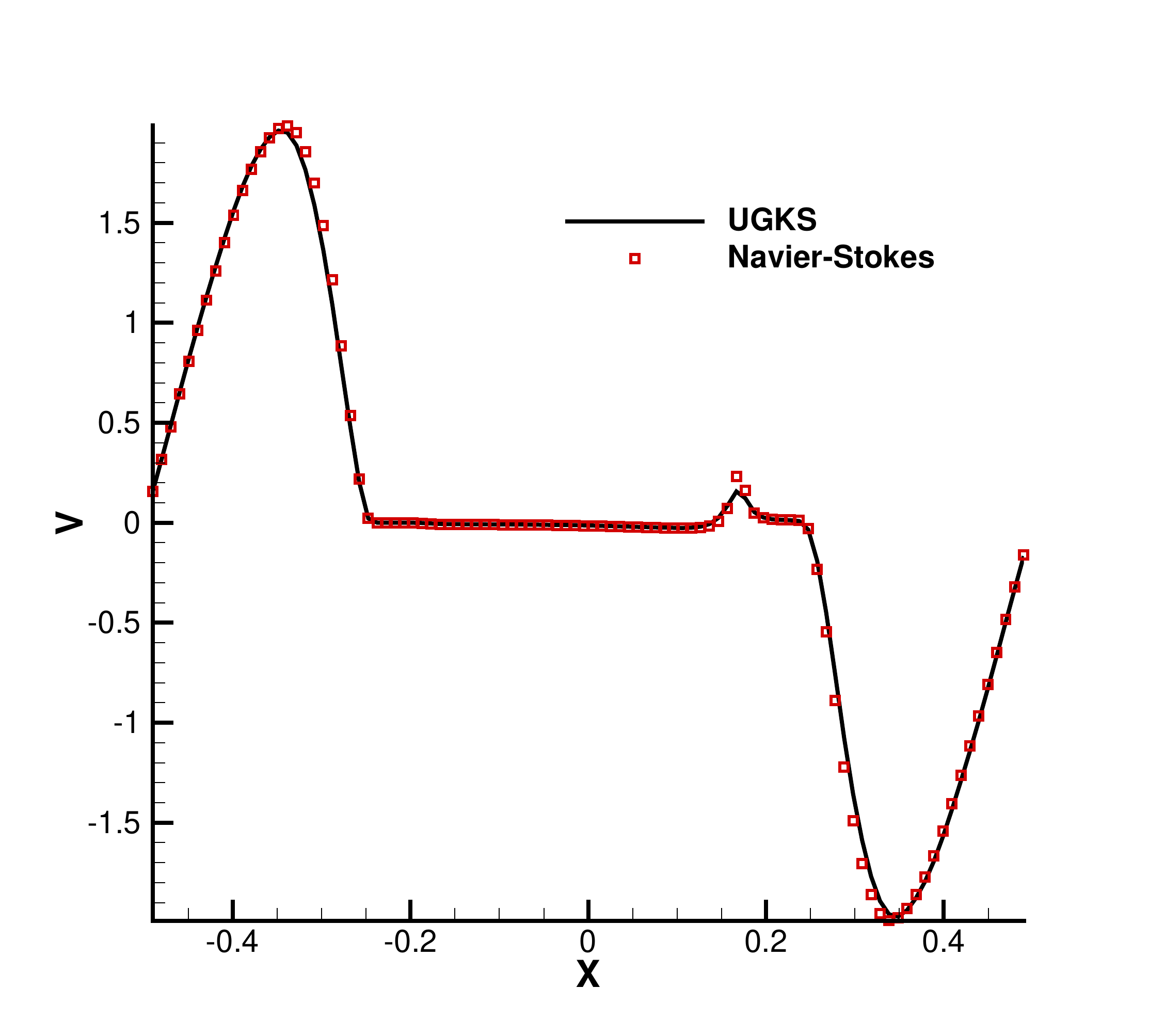}
\includegraphics[width=0.48\textwidth]{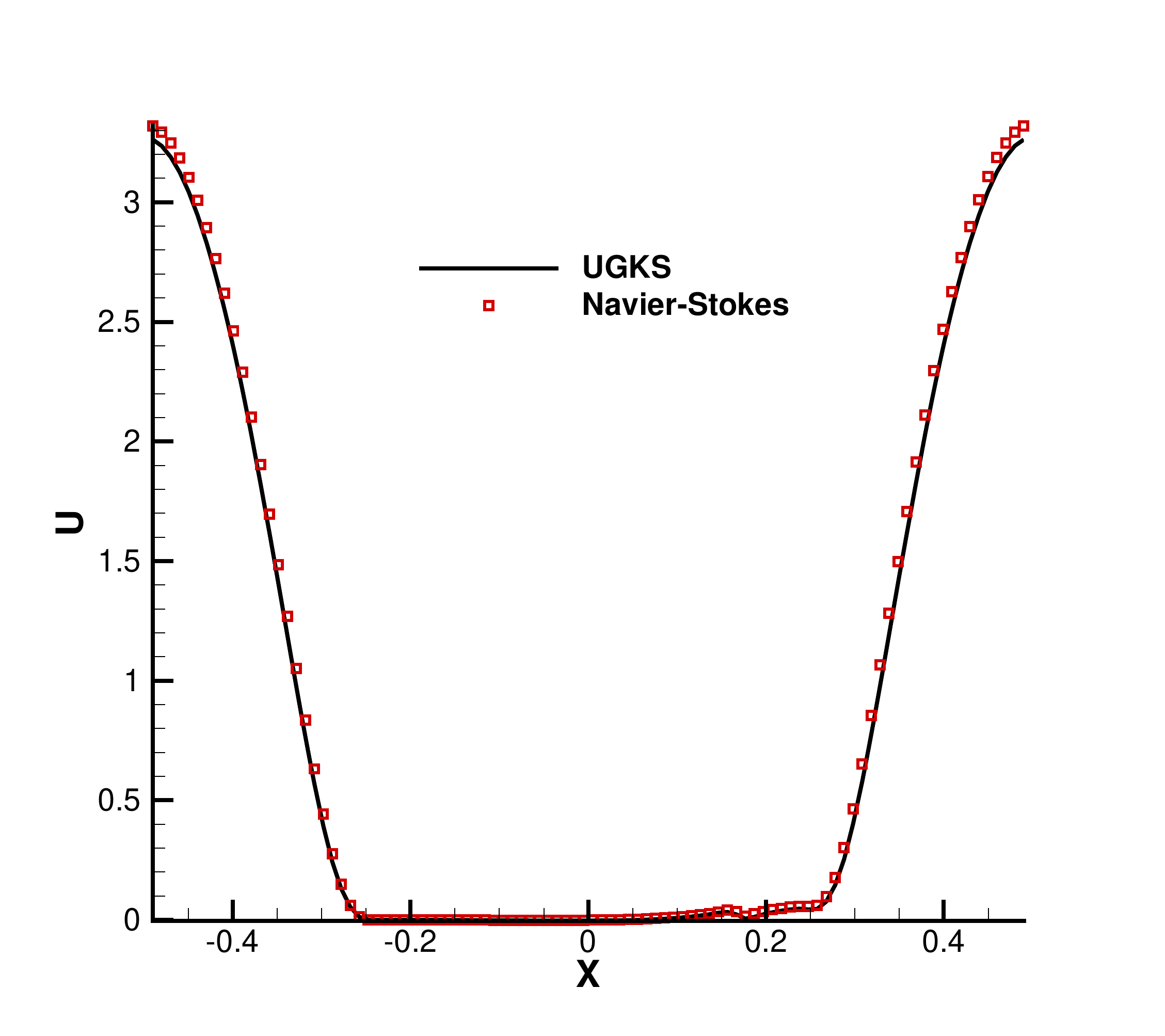}
\caption{The comparison of UGKS and NS velocity profile along $y=0.25$ for the micro flow through periodic square cylinders.}
\label{square3}
\end{figure}

\begin{figure}
\centering
\includegraphics[width=0.48\textwidth]{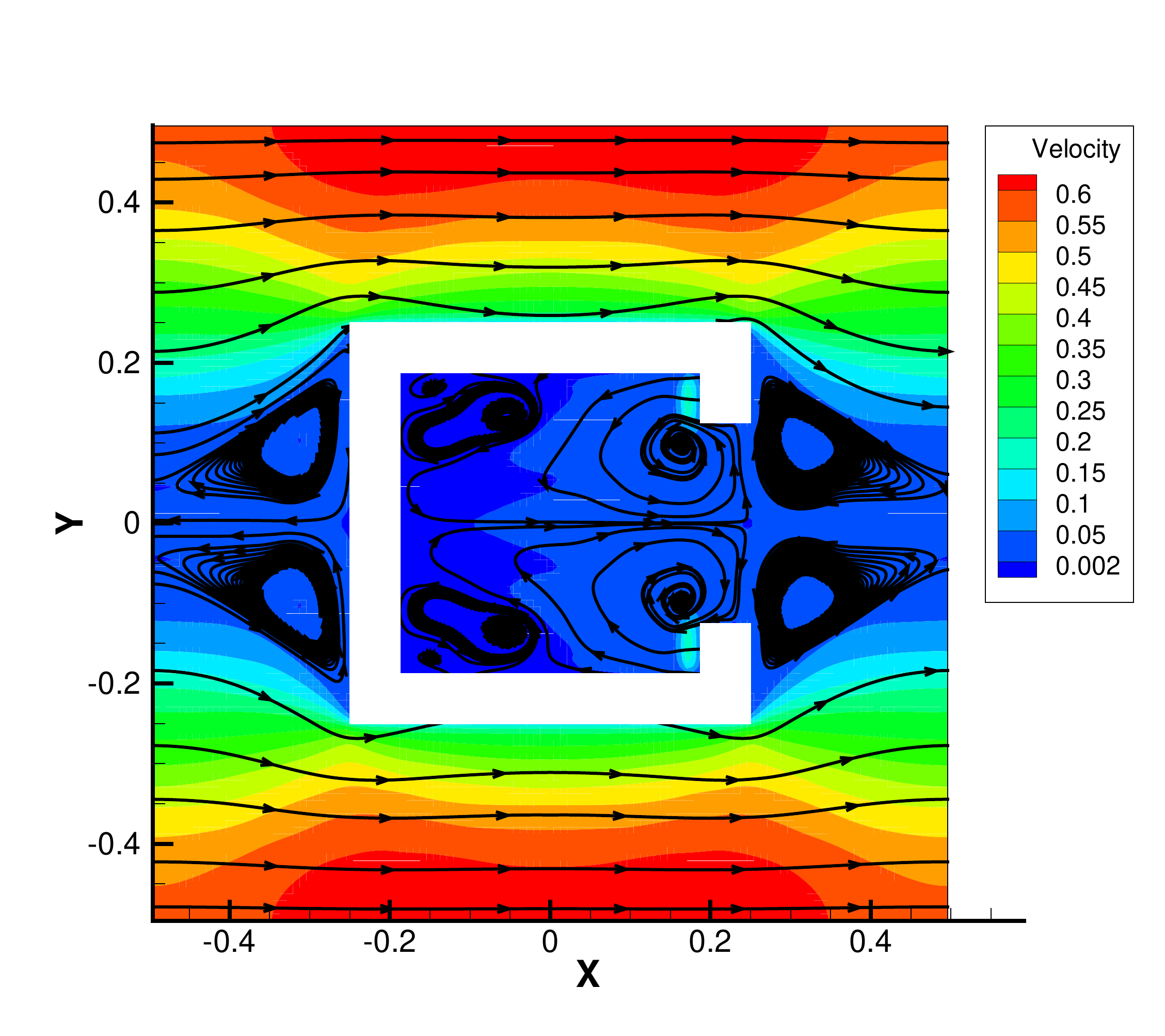}
\includegraphics[width=0.48\textwidth]{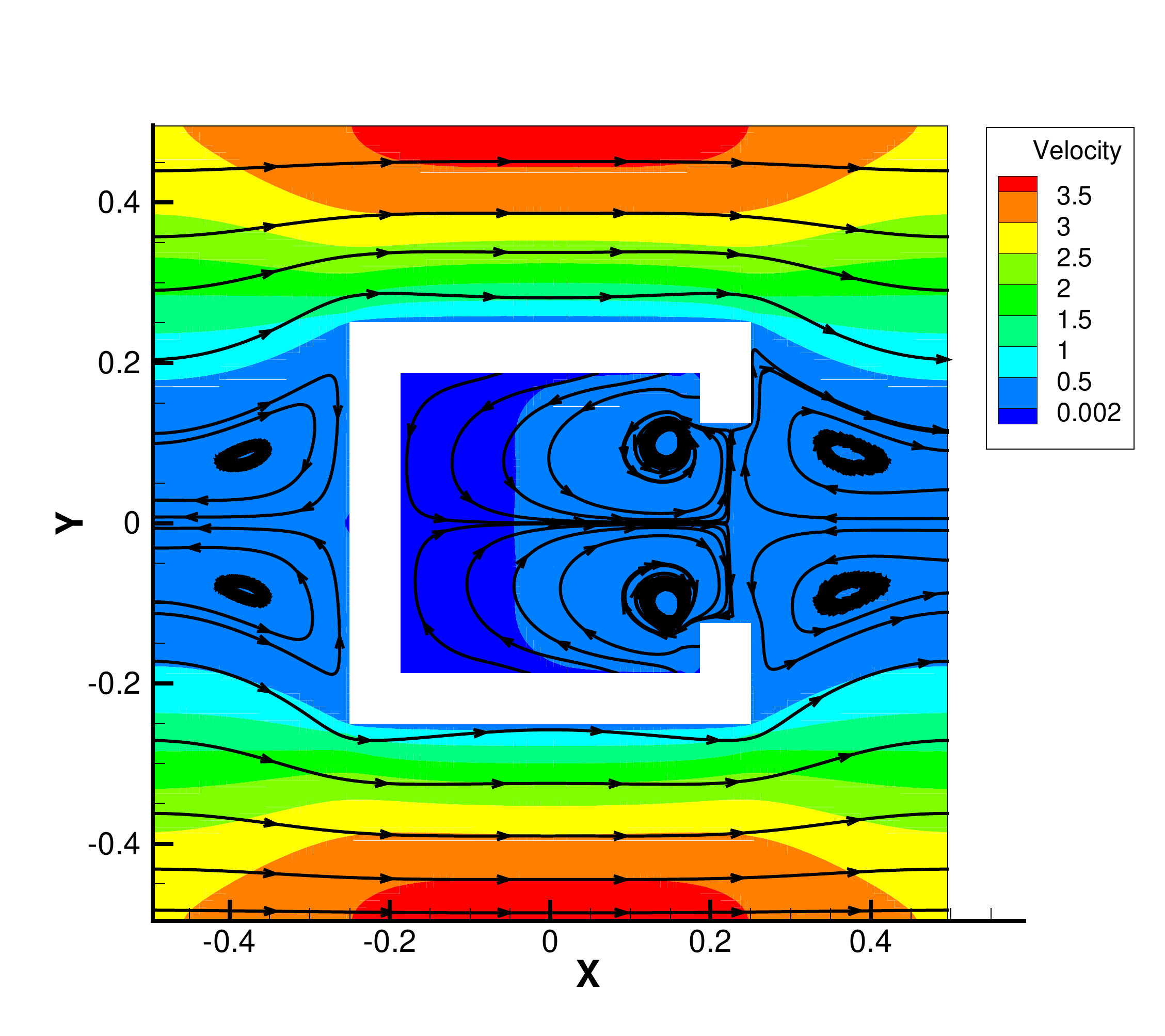}
\caption{The streamline and velocity magnitude of the micro flow through periodic square cylinders. Left figure shows the solution with Knudsen number $\text{Kn}=1.0\times10^{-1}$, and right figure shows the solution with Knudsen number $\text{Kn}=1.0\times10^{-2}$.}
\label{square4}
\end{figure}

\begin{figure}
\centering
\includegraphics[width=0.48\textwidth]{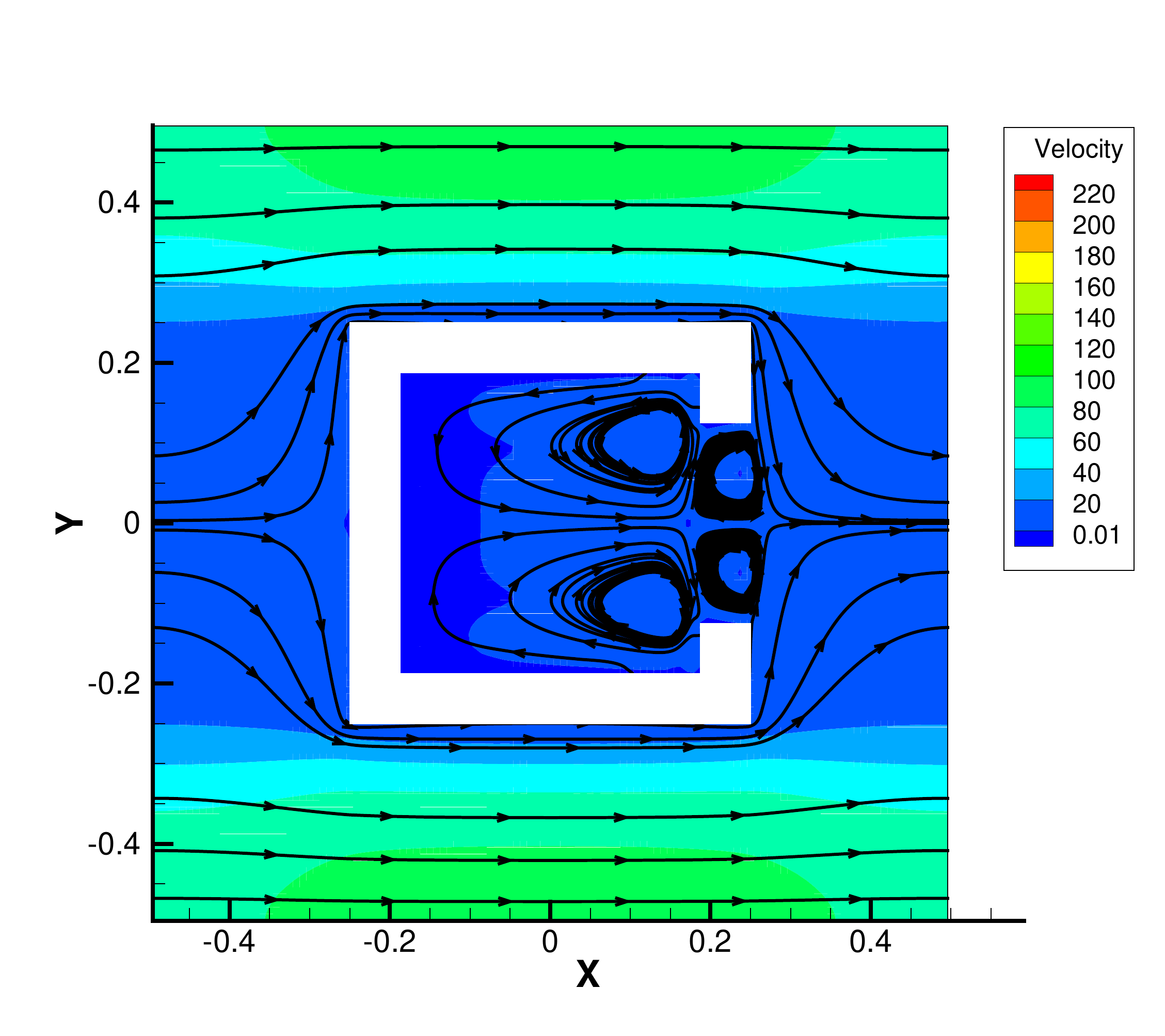}
\includegraphics[width=0.48\textwidth]{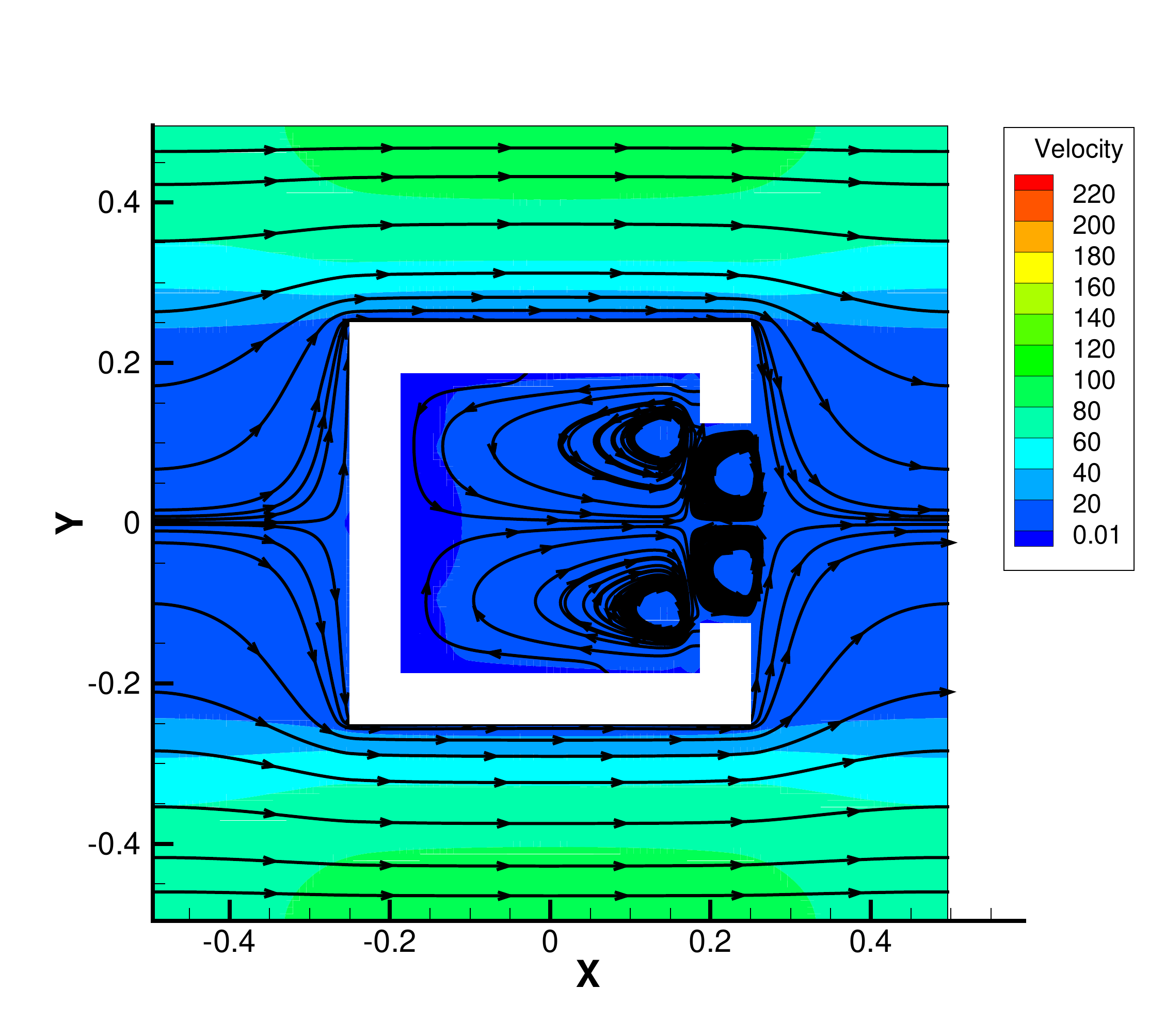}
\caption{The streamline and velocity magnitude of the micro flow through periodic square cylinders with Knudsen number $\text{Kn}=1.0\times10^{-4}$. Left figure shows UGKS solution and right figure shows the NS solution by GKS.}
\label{square5}
\end{figure}

\begin{figure}
\centering
\includegraphics[width=0.48\textwidth]{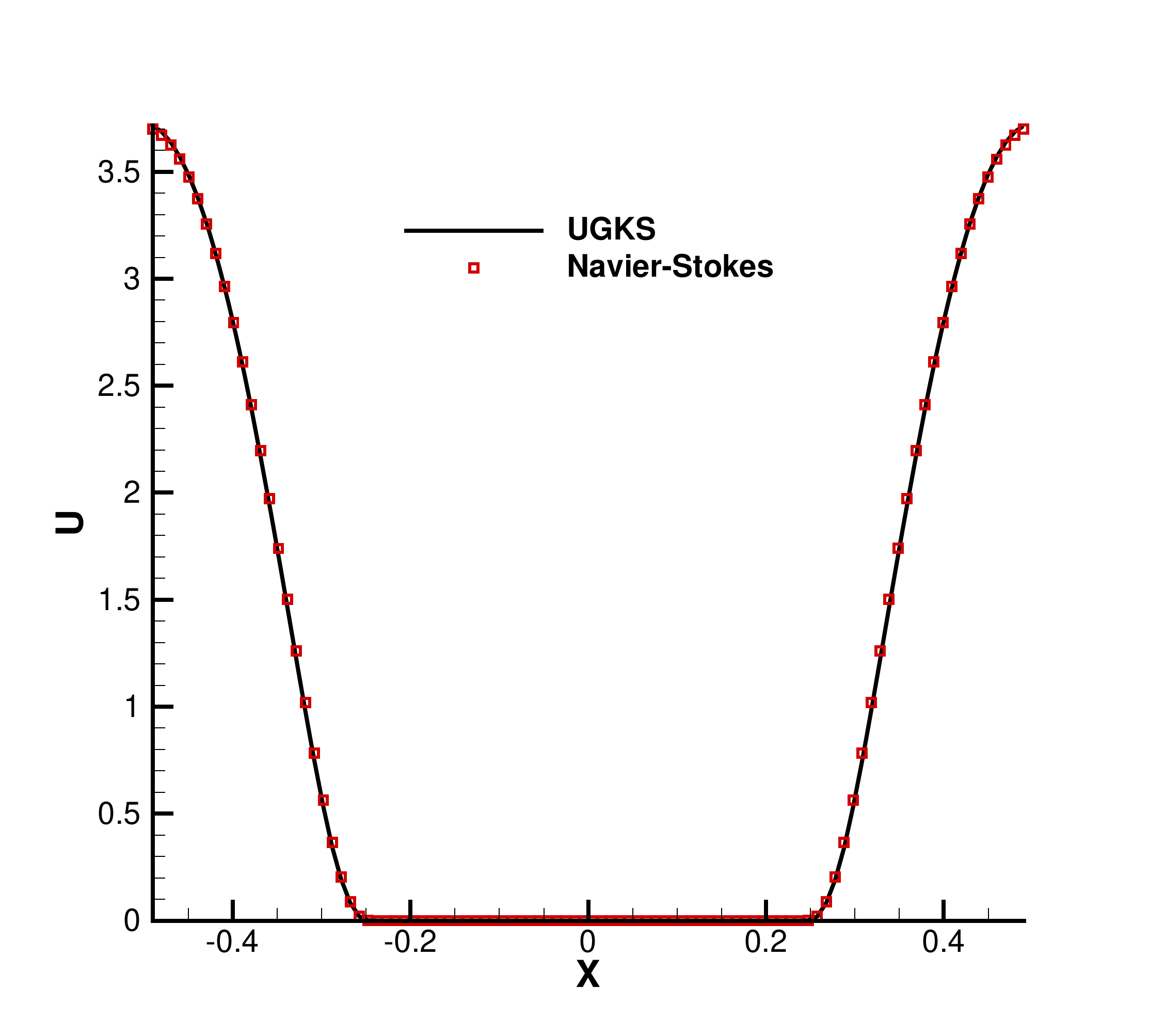}
\includegraphics[width=0.48\textwidth]{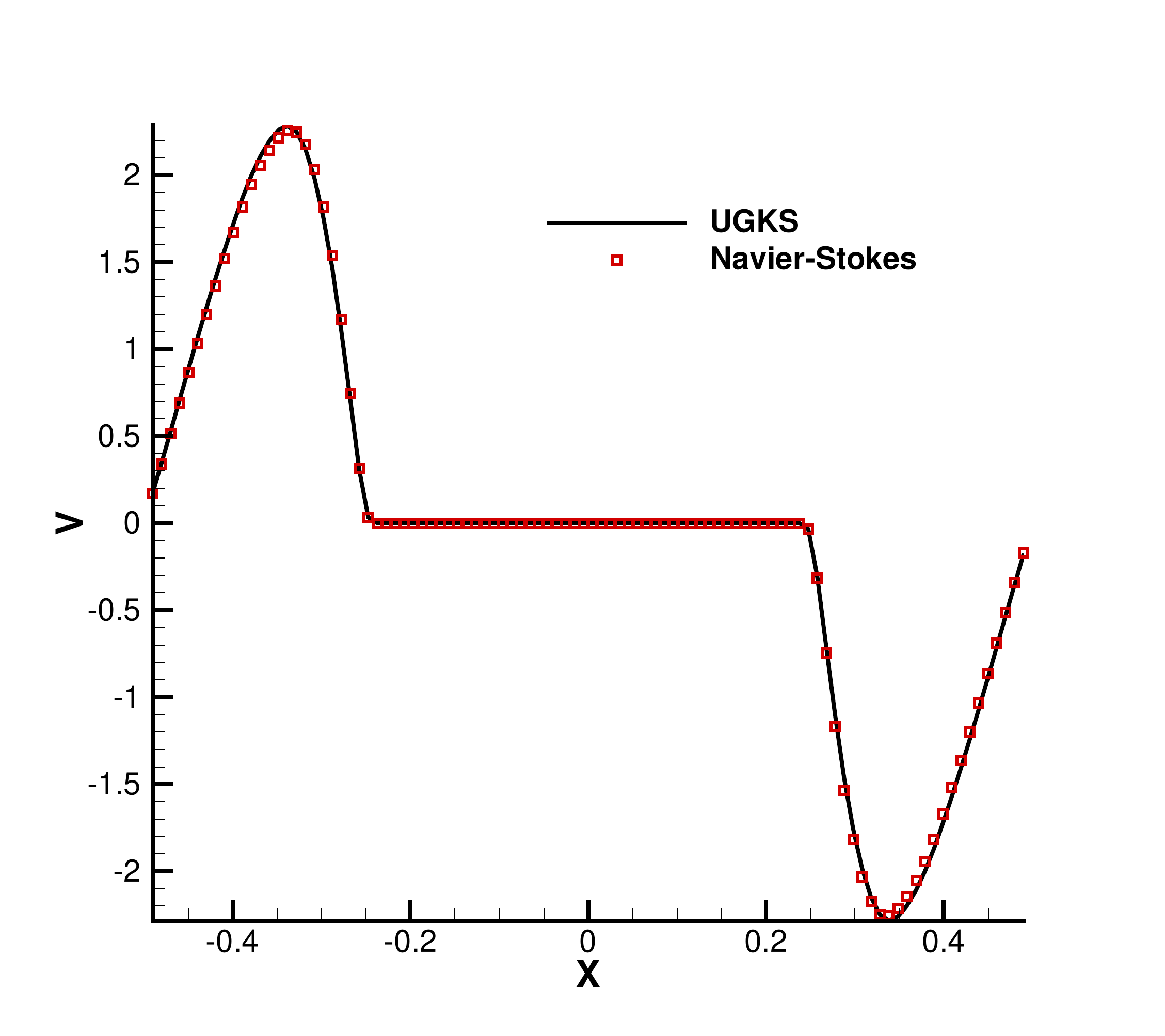}
\caption{The comparison of UGKS and NS velocity profile along $y=0.25$ for the micro flow through periodic square cylinders.}
\label{square6}
\end{figure}

\begin{figure}
\centering
\includegraphics[width=0.48\textwidth]{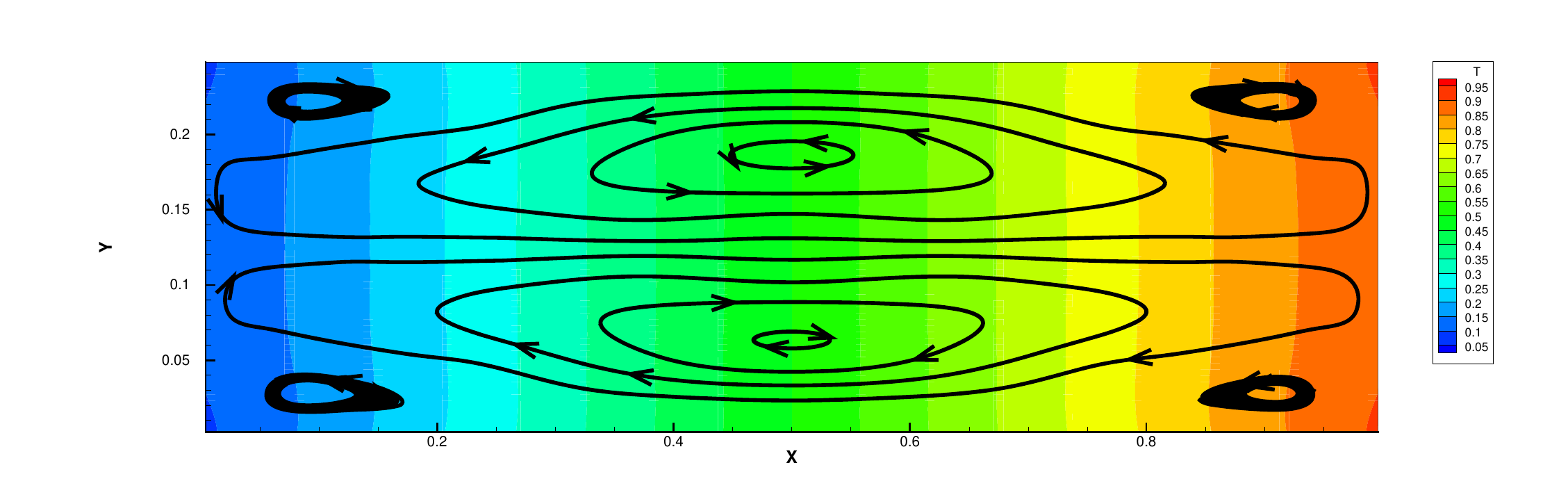}{a}
\includegraphics[width=0.48\textwidth]{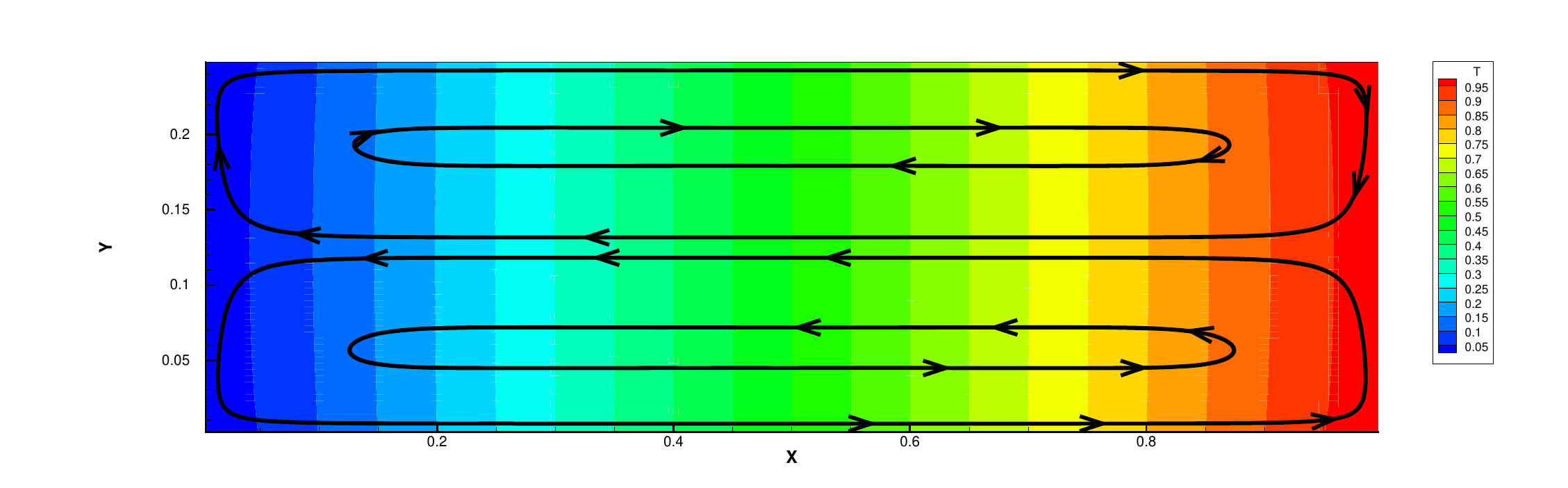}{b}
\caption{The streamline and temperature distribution of the thermal creep flow. Left figure shows the solution with Knudsen number $\text{Kn}=10$, and right figure shows the solution with Knudsen number $\text{Kn}=1.0\times10^{-2}$.}
\label{creep1}
\end{figure}

\begin{figure}
\centering
\includegraphics[width=0.48\textwidth]{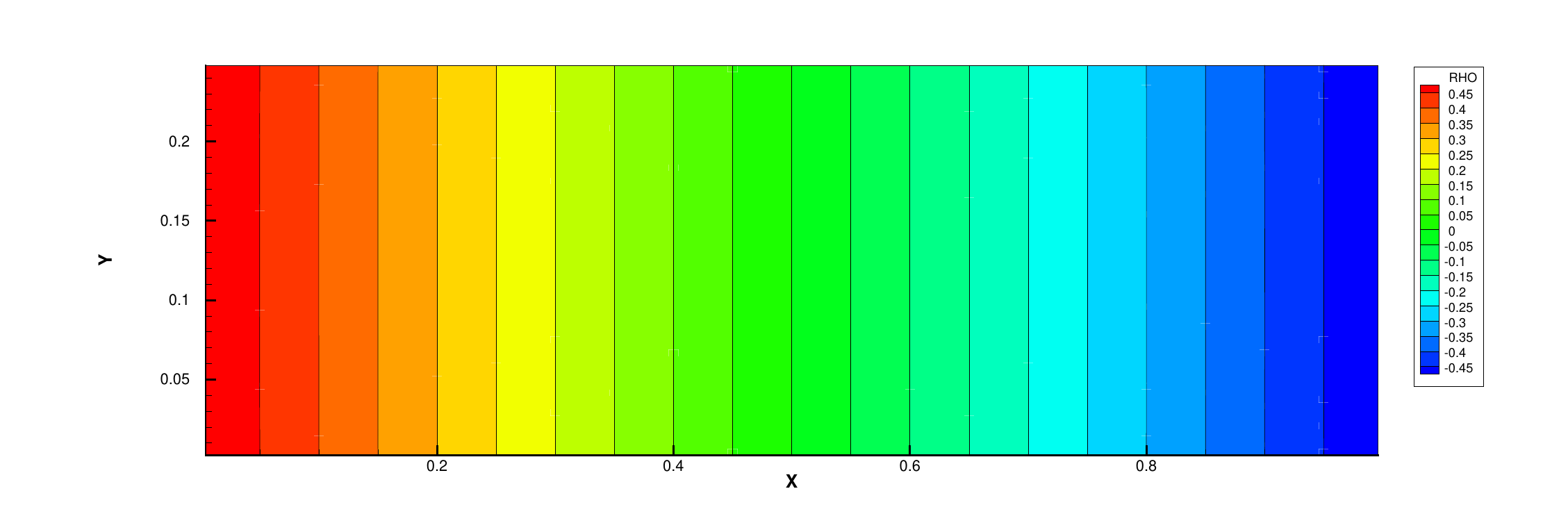}{a}
\includegraphics[width=0.48\textwidth]{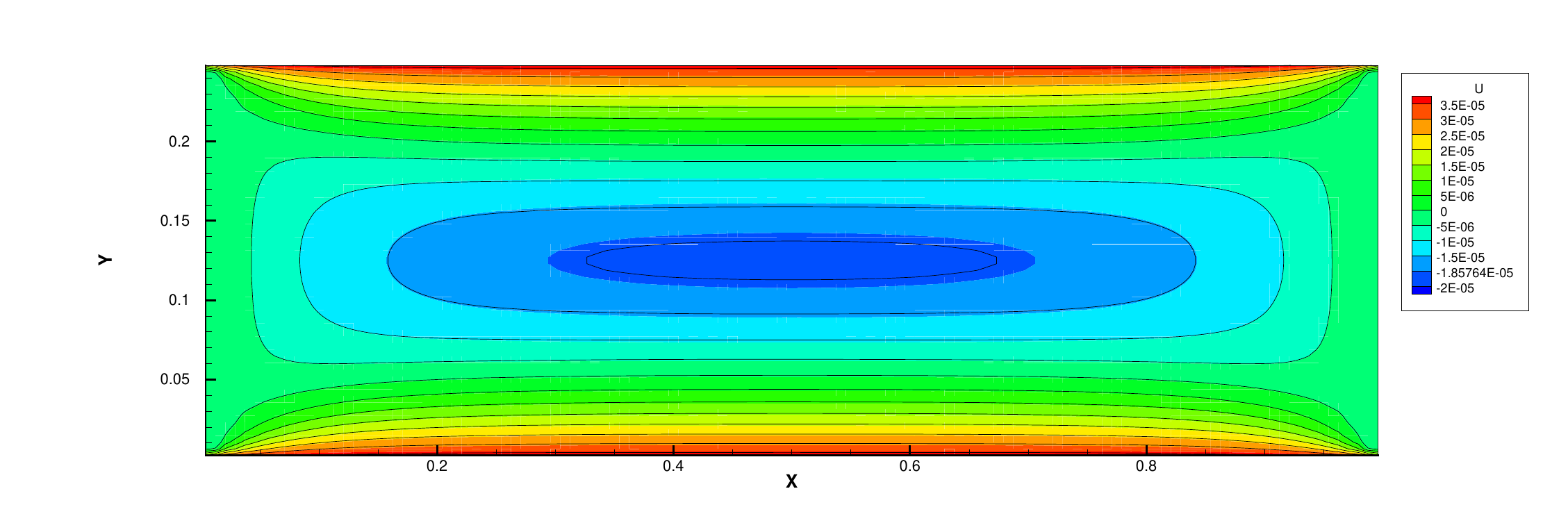}{b}
\includegraphics[width=0.48\textwidth]{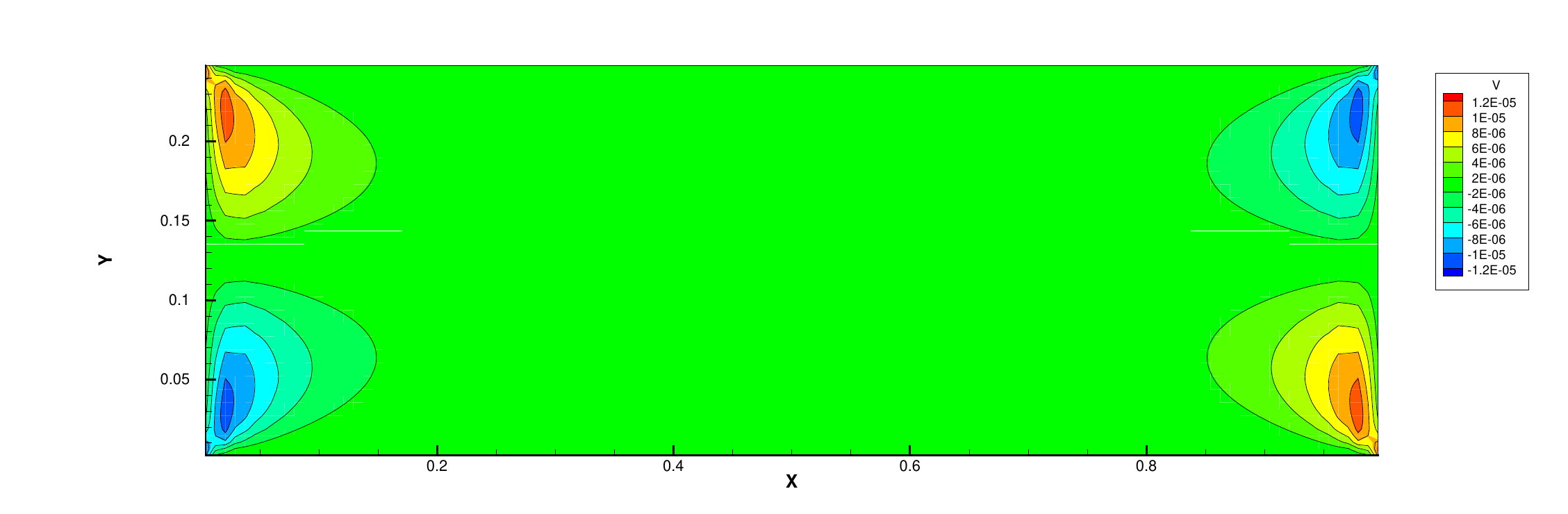}{c}
\includegraphics[width=0.48\textwidth]{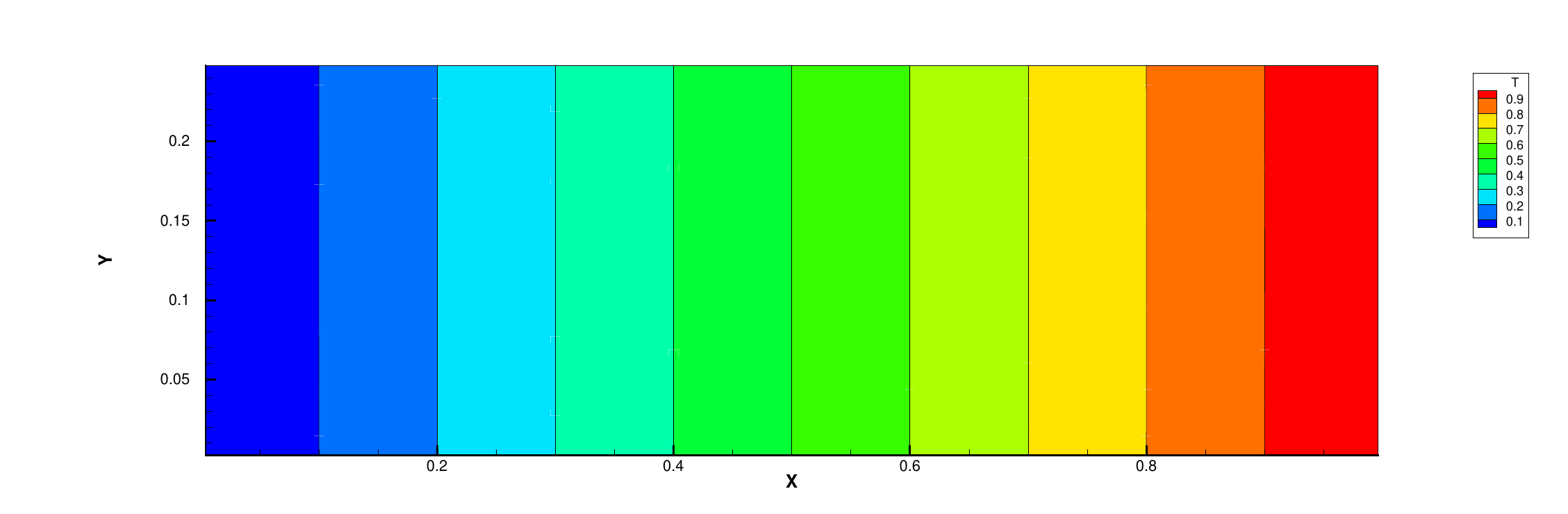}{d}
\caption{The comparison of UGKS (contour) and NS (solid line) solution for the thermal creep flow. (a) density distribution; (b) x-directional velocity distributionl; (c) y-directional velocity distributionl; (d) temperature distribution.}
\label{creep2}
\end{figure}

\begin{figure}
\centering
\includegraphics[width=0.48\textwidth]{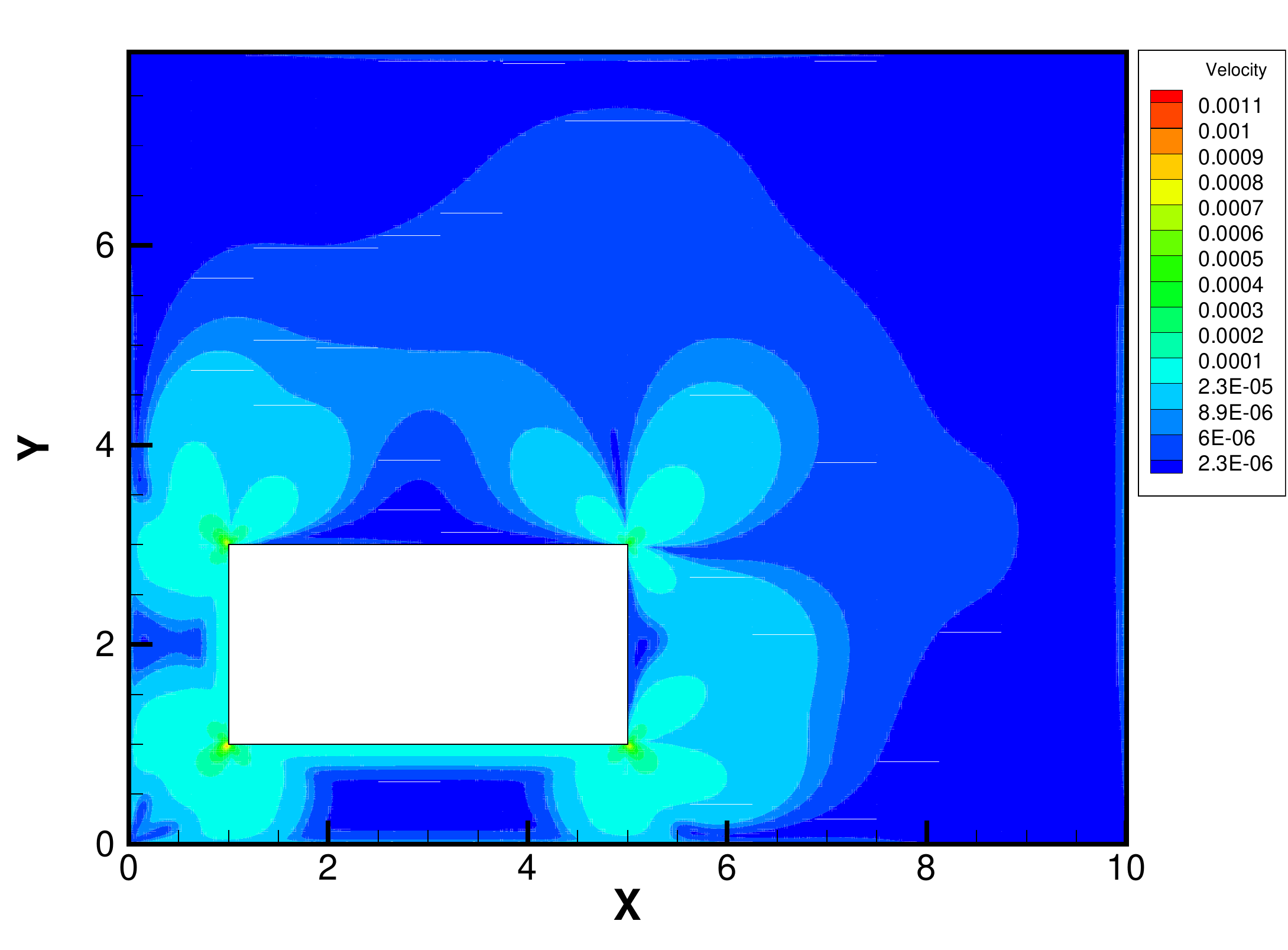}{a}
\includegraphics[width=0.48\textwidth]{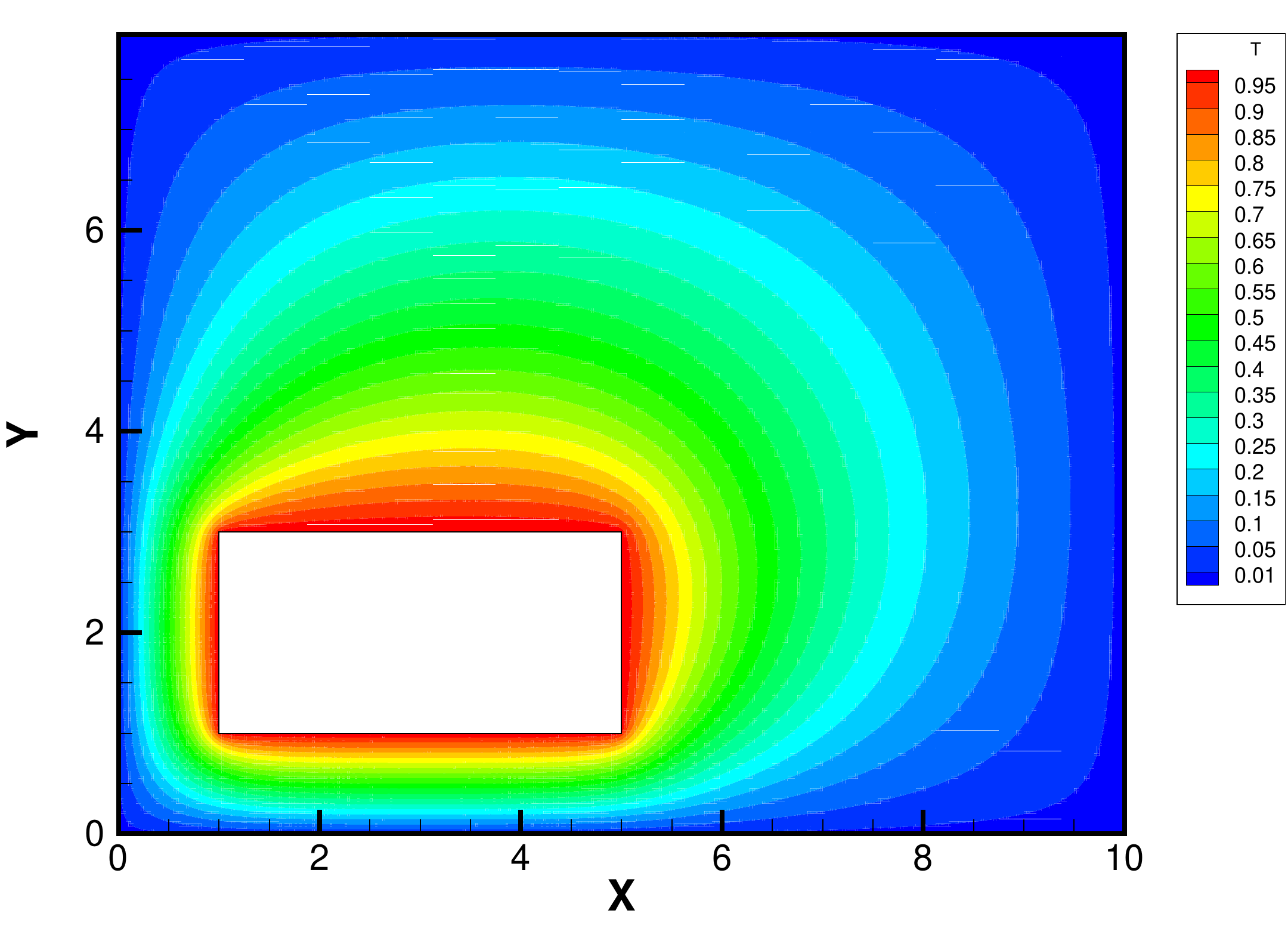}{b}
\caption{The velocity magnitude (left) and temperature distribution (right) of the microbeam flow.}
\label{microbeam1}
\end{figure}

\begin{figure}
\centering
\includegraphics[width=0.5\textwidth]{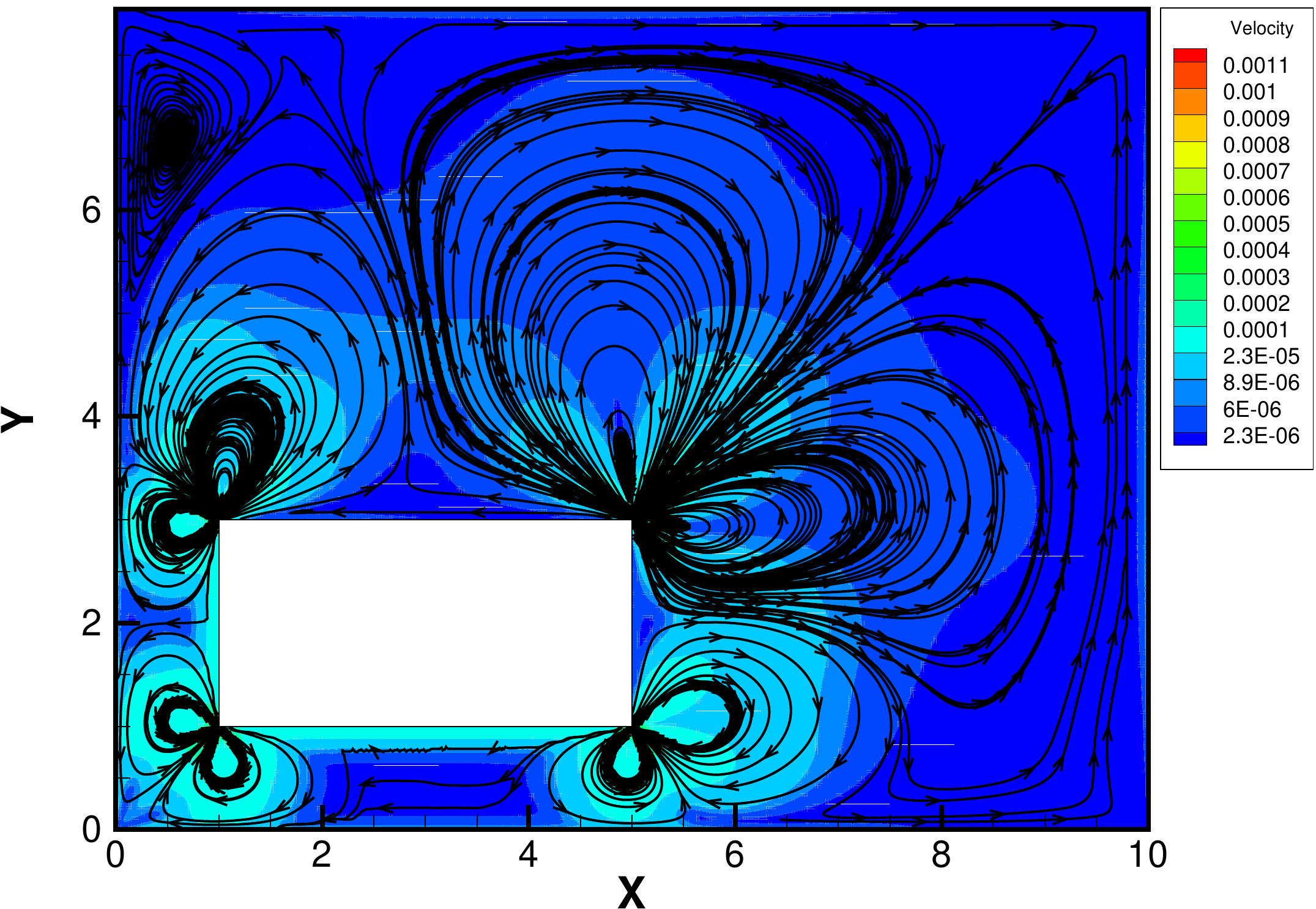}
\includegraphics[width=0.48\textwidth]{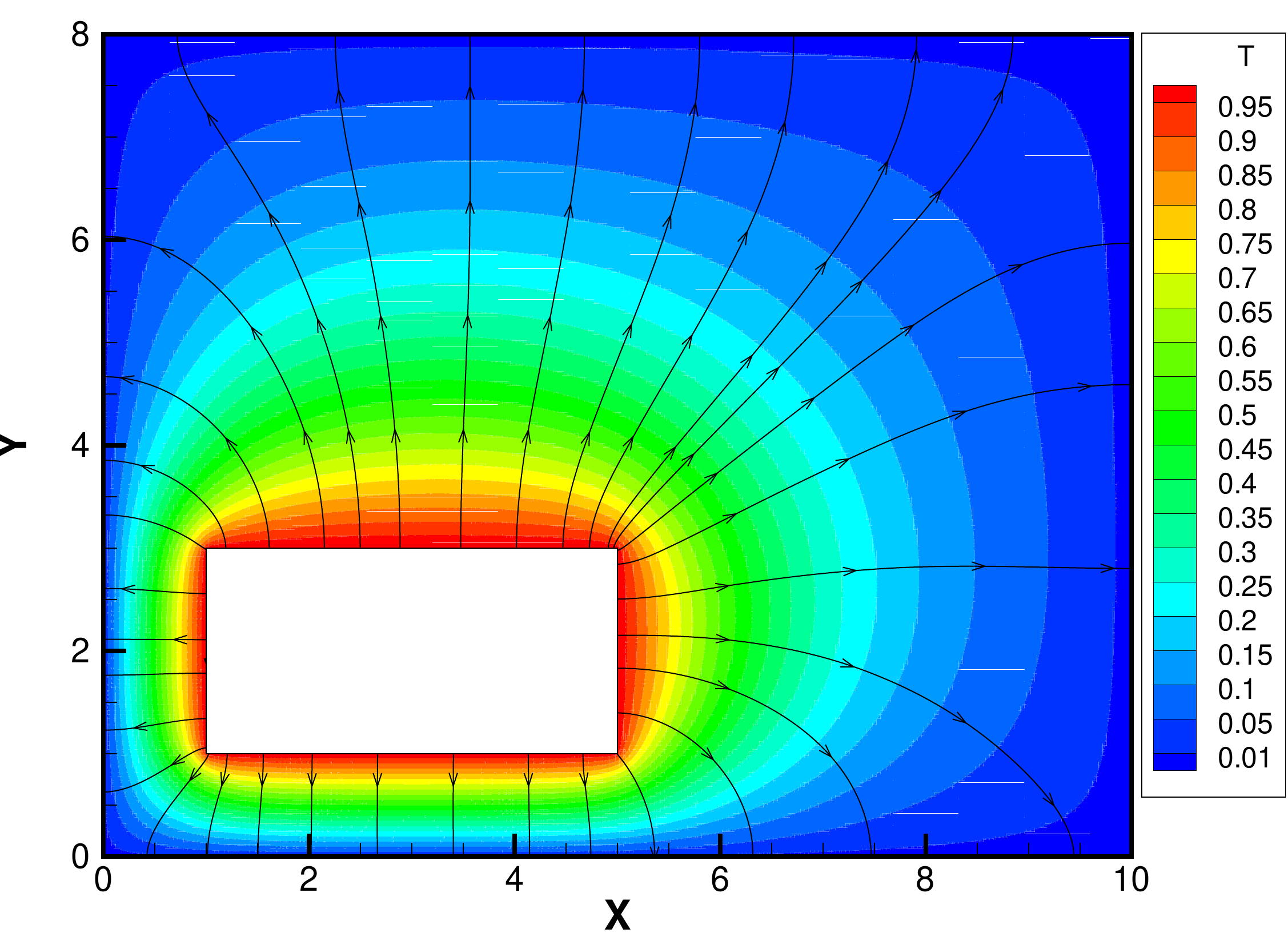}
\caption{Left figure shows the streamline of the microbeam flow with the velocity magnitude background, and right figure shows the heat flow of the microbeam flow with the temperature magnitude background.}
\label{microbeam2}
\end{figure}

\begin{figure}
\centering
\includegraphics[width=0.48\textwidth]{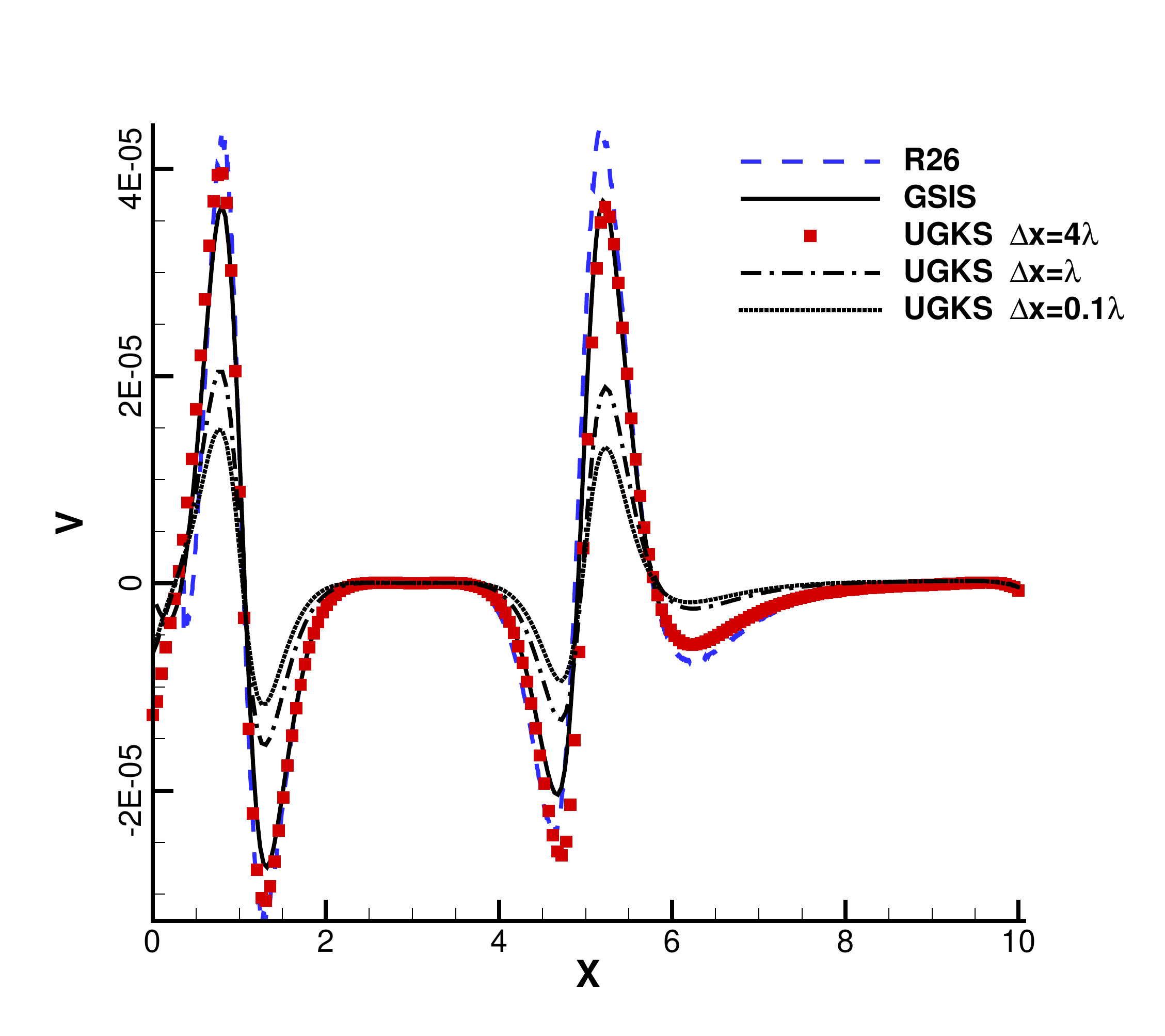}
\includegraphics[width=0.48\textwidth]{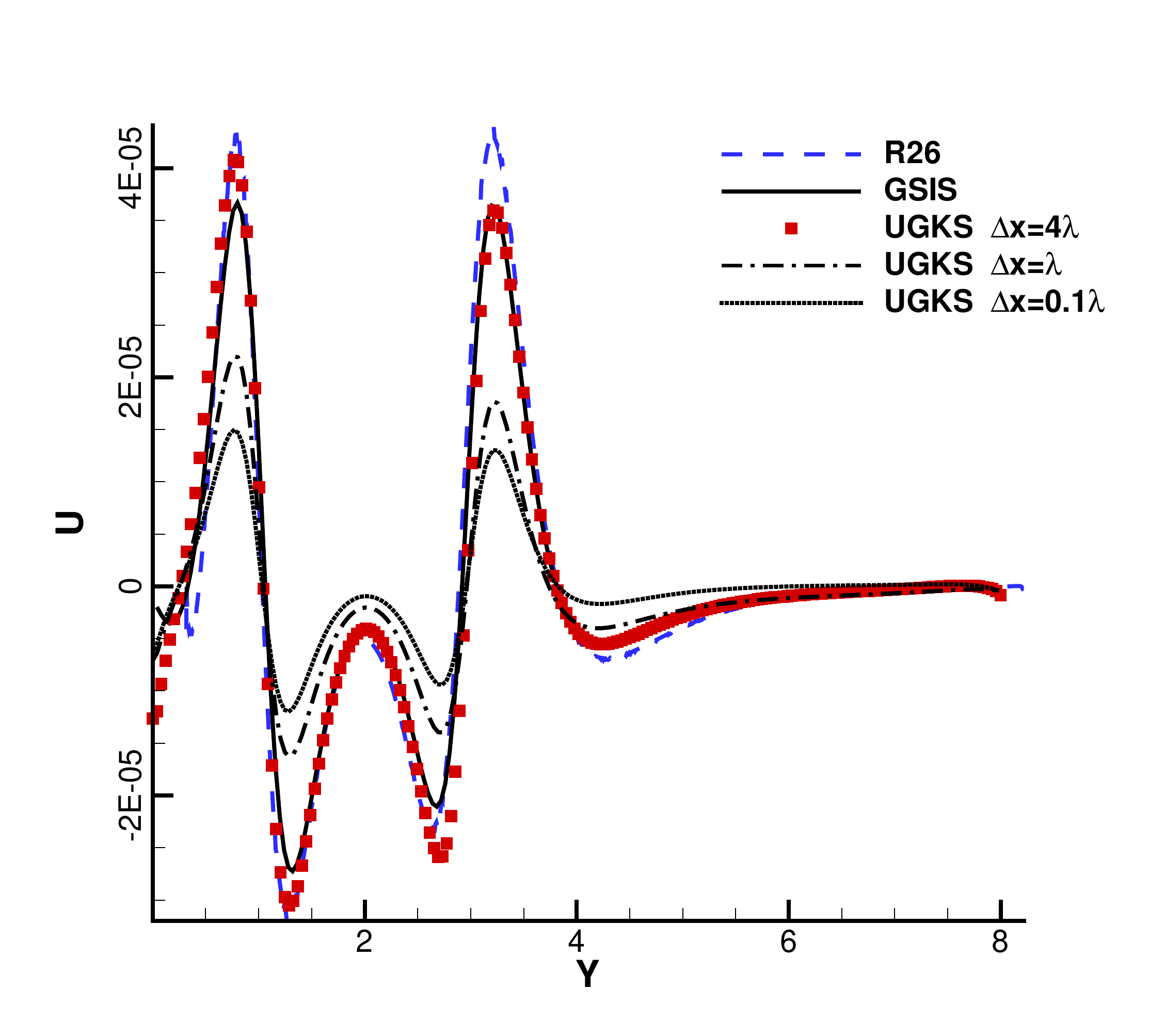}
\caption{Left figure shows the y-direction velocity profile along $y=0.5$, and right figure shows the x-direction velocity profile along $x=0.5$. The UGKS result with $\Delta x=4\lambda$ is shown in symbol; the UGKS result with $\Delta x=\lambda$ is shown in dash-dotted line; the UGKS result with $\Delta x=0.1\lambda$ is shown in dotted line; the R26 solution is shown in dashed line; and the GSIS solution is shown in solid line.}
\label{microbeam3}
\end{figure}

\begin{figure}
\centering
\includegraphics[width=0.48\textwidth]{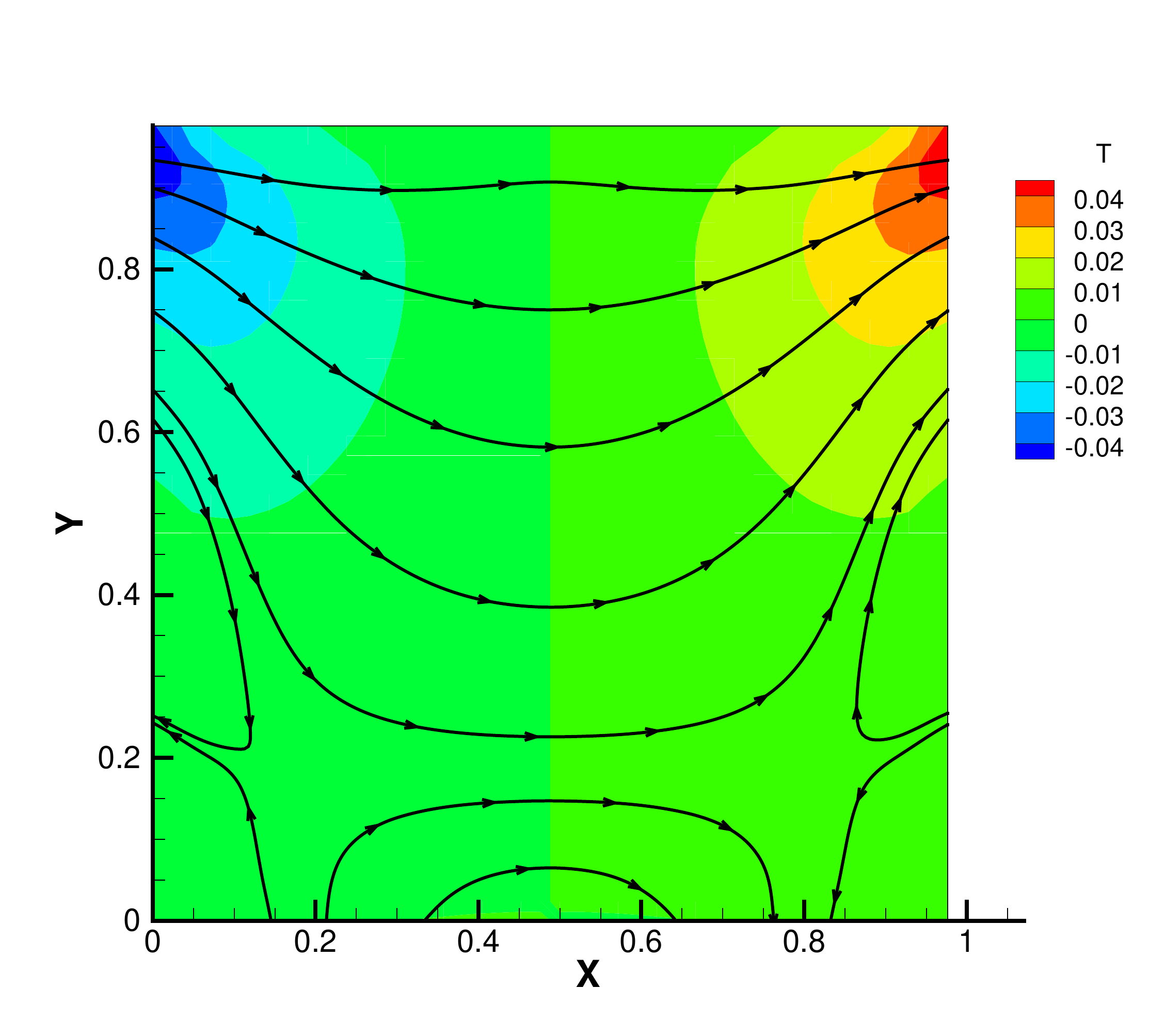}
\includegraphics[width=0.48\textwidth]{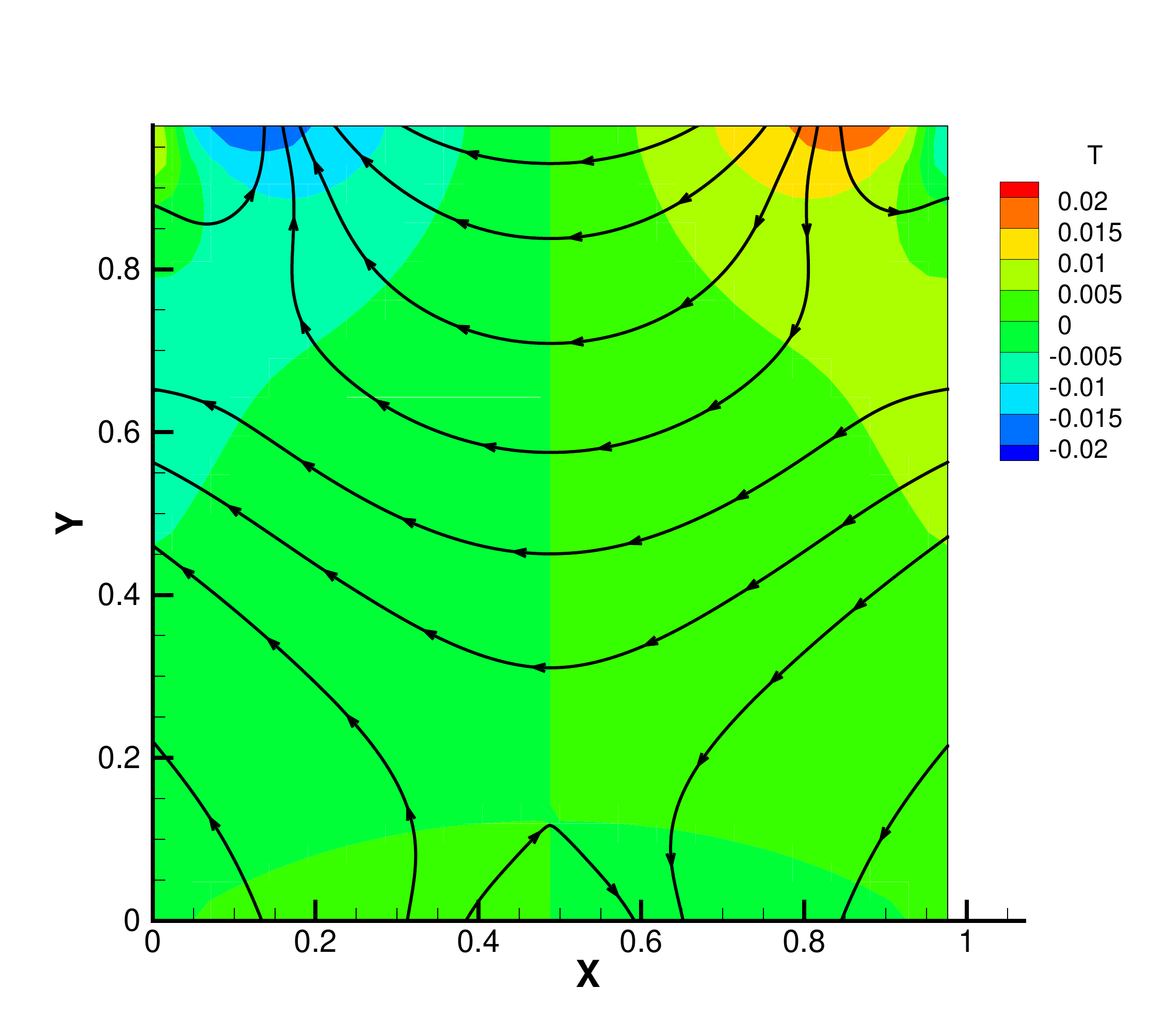}
\caption{The heat flux and temperature contour of the lid-driven cavity flow with Knudsen number 0.1. Left figure shows the UGKS solution and right figure shows the Navier-Stokes solution.}
\label{cavity1}
\end{figure}

\begin{figure}
\centering
\includegraphics[width=0.48\textwidth]{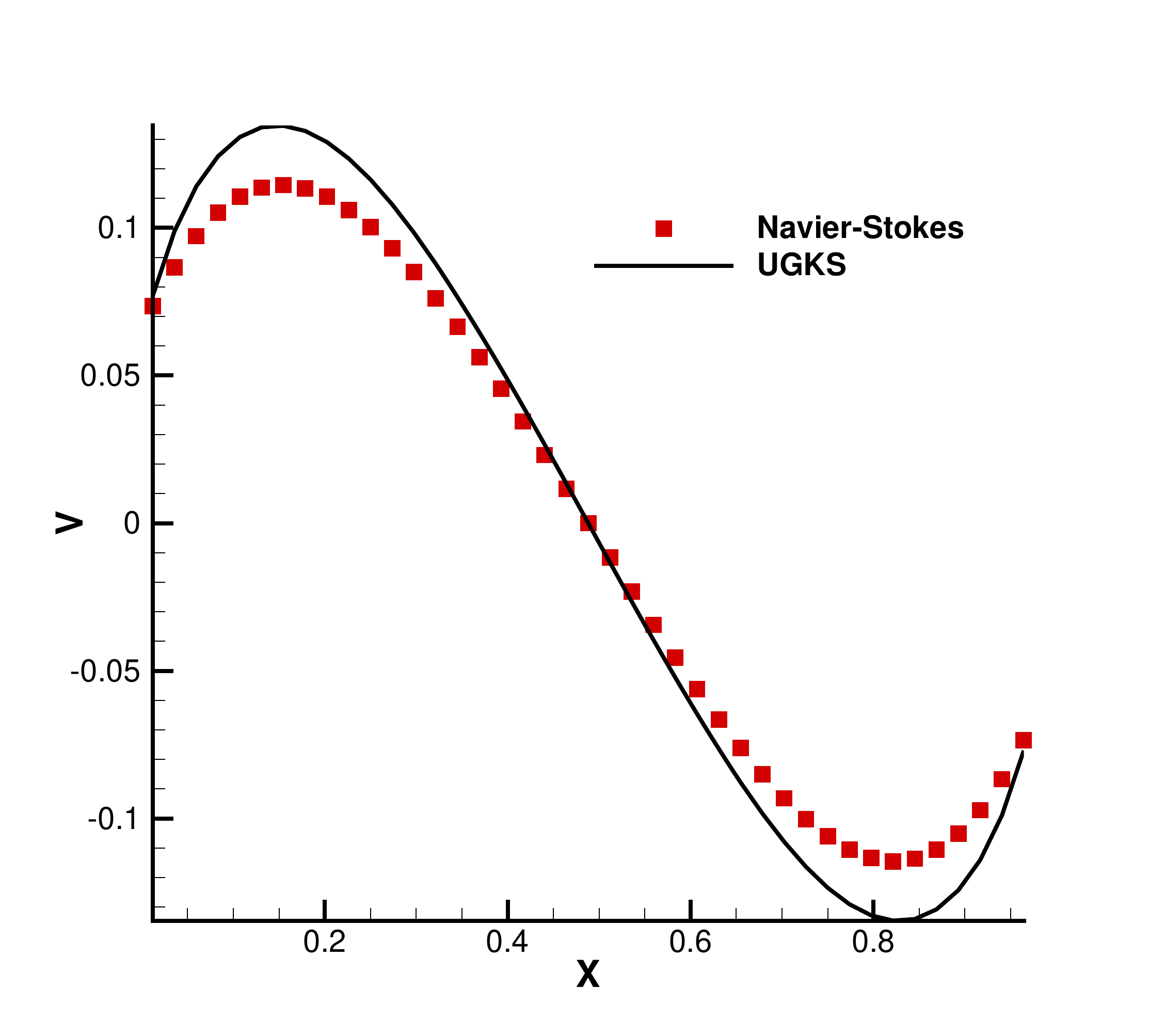}
\includegraphics[width=0.48\textwidth]{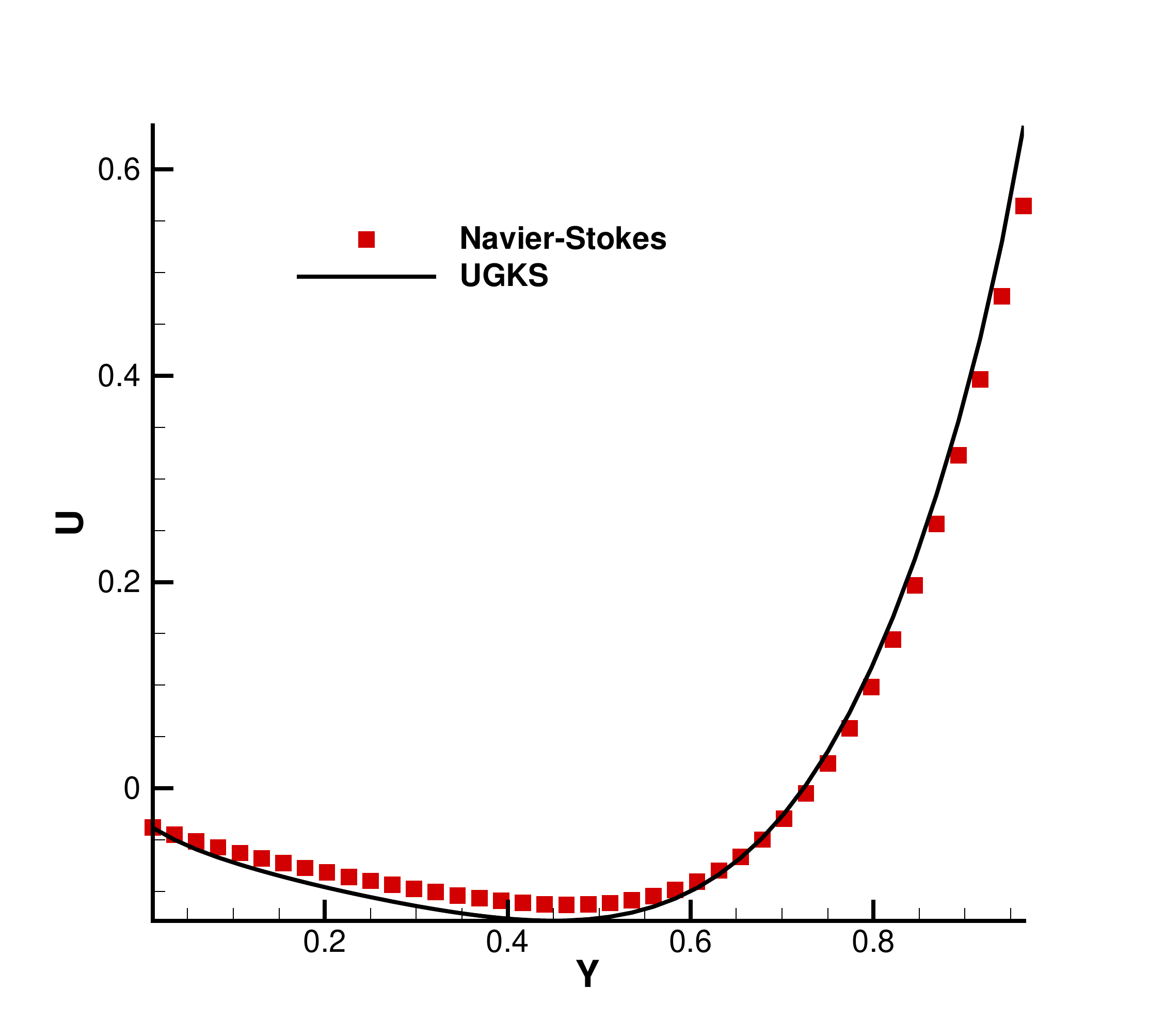}
\caption{The velocity distribution of the lid-driven cavity flow with Knudsen number 0.1. Left figure shows the y-directional velocity distribution along $y=0.5$, and right figure shows the x-directional velocity distribution along $x=0.5$.}
\label{cavity2}
\end{figure}

\begin{figure}
\centering
\includegraphics[width=0.48\textwidth]{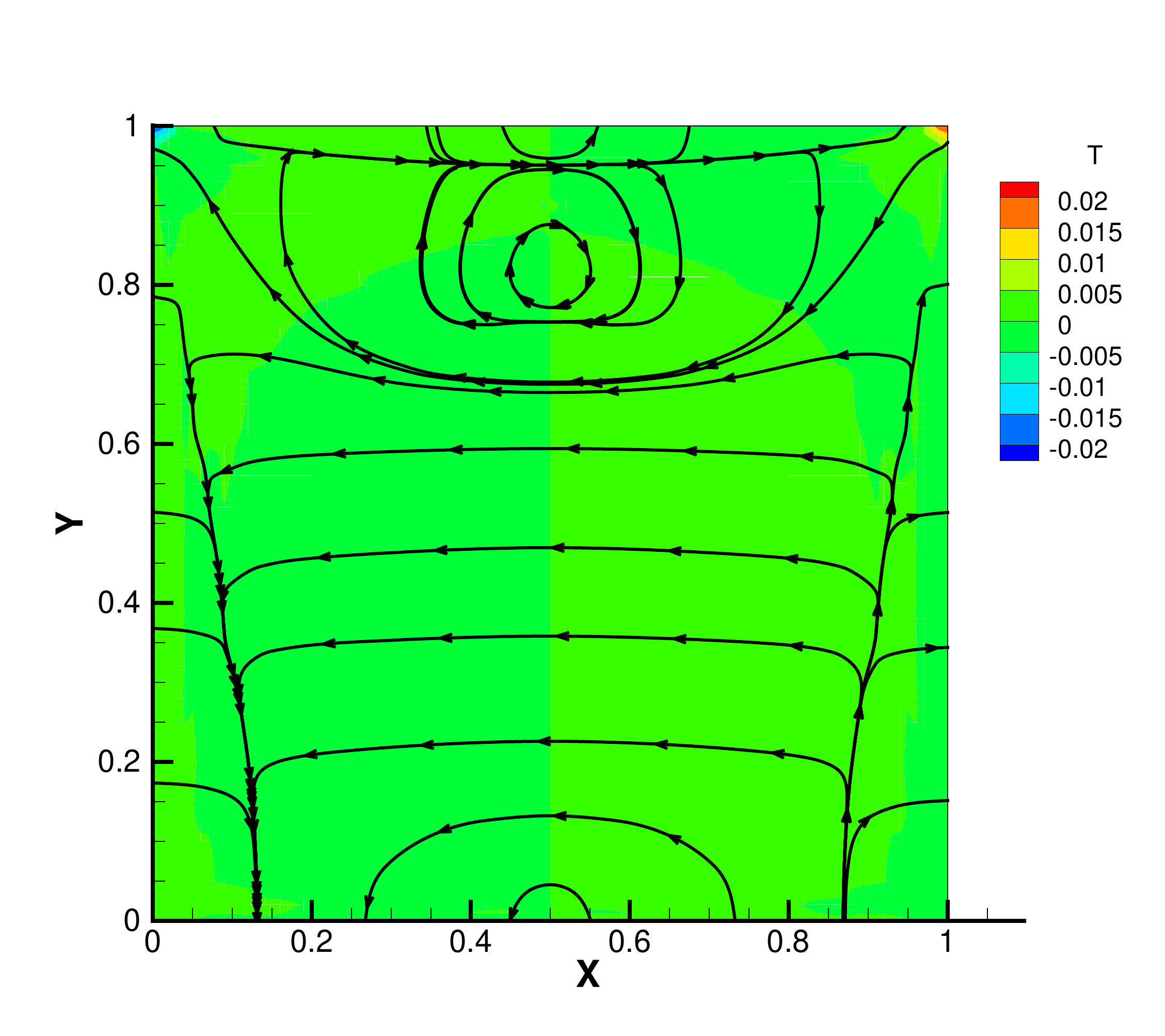}
\includegraphics[width=0.48\textwidth]{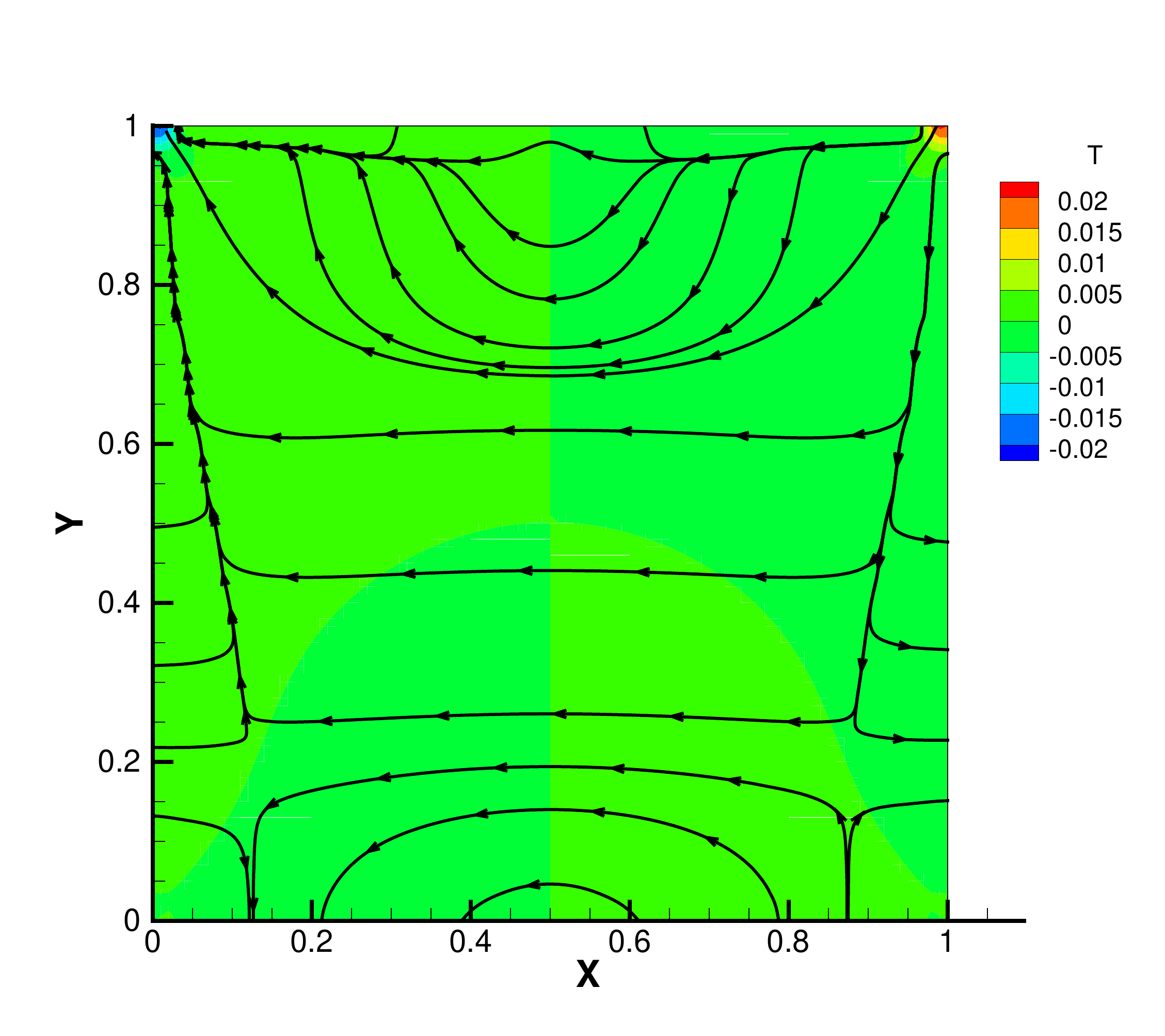}
\caption{The heat flux and temperature contour of the lid-driven cavity flow with Knudsen number $10^{-4}$. Left figure shows the UGKS solution and right figure shows the Navier-Stokes solution.}
\label{cavity3}
\end{figure}

\begin{figure}
\centering
\includegraphics[width=0.48\textwidth]{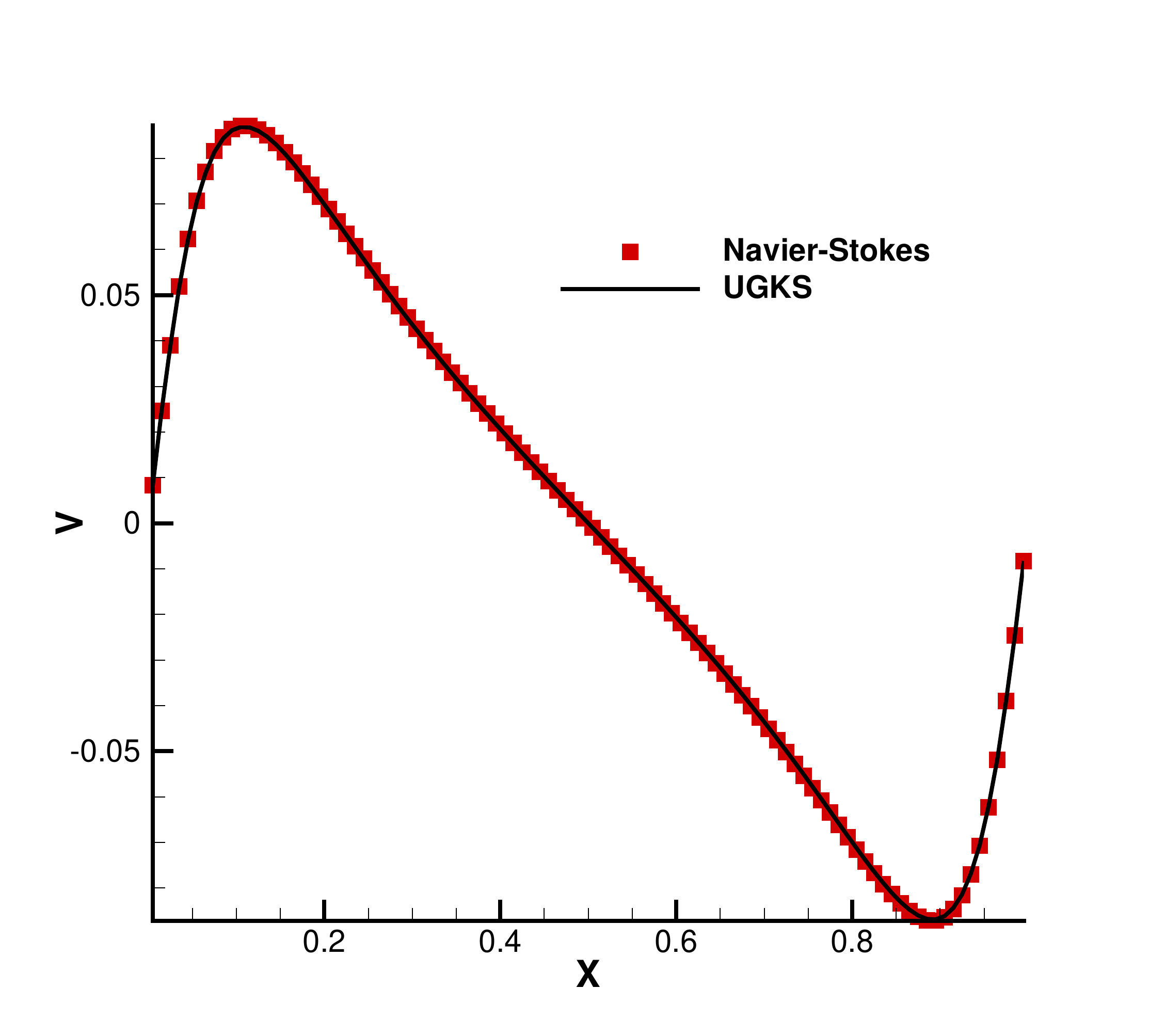}
\includegraphics[width=0.48\textwidth]{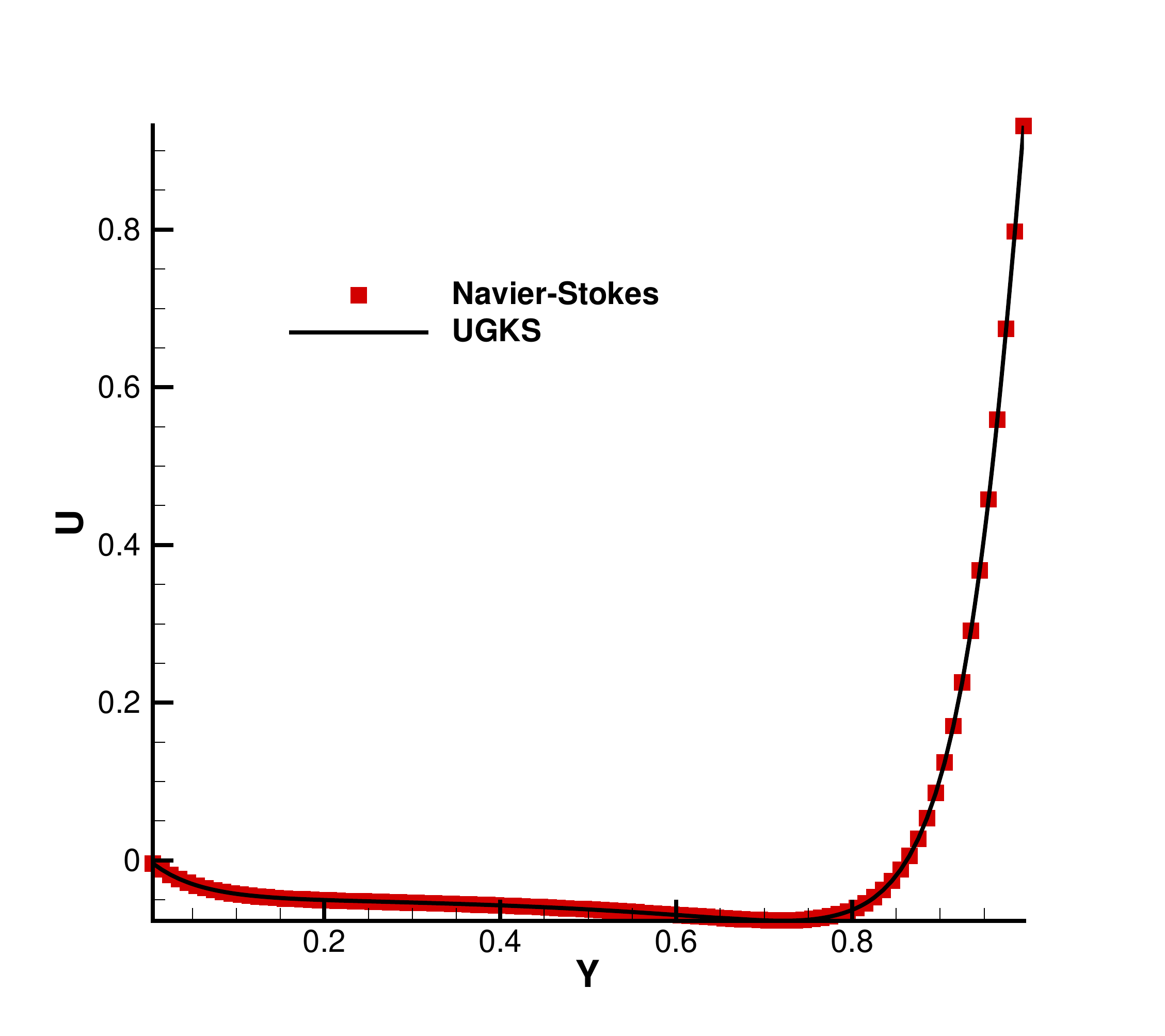}
\caption{The velocity distribution of the lid-driven cavity flow with Knudsen number $10^{-4}$. Left figure shows the y-directional velocity distribution along $y=0.5$, and right figure shows the x-directional velocity distribution along $x=0.5$.}
\label{cavity4}
\end{figure}

\section*{References}
\bibliography{linear}

\begin{thebibliography}{10}
\expandafter\ifx\csname url\endcsname\relax
  \def\url#1{\texttt{#1}}\fi
\expandafter\ifx\csname urlprefix\endcsname\relax\def\urlprefix{URL }\fi
\expandafter\ifx\csname href\endcsname\relax
  \def\href#1#2{#2} \def\path#1{#1}\fi

\bibitem{guo2019unified}
Z.~Guo, J.~Li, K.~Xu, On unified preserving properties of kinetic schemes,
  arXiv preprint arXiv:1909.04923.

\bibitem{chapman1990mathematical}
S.~Chapman, T.~G. Cowling, D.~Burnett, The mathematical theory of non-uniform
  gases: {An} account of the kinetic theory of viscosity, thermal conduction
  and diffusion in gases, Cambridge university press, 1990.

\bibitem{cercignani1969mathematical}
C.~Cercignani, Mathematical methods in kinetic theory, Springer, 1969.

\bibitem{ugks2010}
K.~Xu, J.-C. Huang, A unified gas-kinetic scheme for continuum and rarefied
  flows, Journal of Computational Physics 229~(20) (2010) 7747--7764.

\bibitem{guo2013discrete}
Z.~Guo, K.~Xu, R.~Wang, Discrete unified gas kinetic scheme for all knudsen
  number flows: Low-speed isothermal case, Physical Review E 88~(3) (2013)
  033305.

\bibitem{su2020can}
W.~Su, L.~Zhu, P.~Wang, Y.~Zhang, L.~Wu, Can we find steady-state solutions to
  multiscale rarefied gas flows within dozens of iterations?, Journal of
  Computational Physics (2020) 109245.

\bibitem{yuan2019multi}
R.~Yuan, C.~Zhong, A multi-prediction implicit scheme for steady state
  solutions of gas flow in all flow regimes, arXiv preprint arXiv:1905.06629.

\bibitem{jenny2010solution}
P.~Jenny, M.~Torrilhon, S.~Heinz, A solution algorithm for the fluid dynamic
  equations based on a stochastic model for molecular motion, Journal of
  computational physics 229~(4) (2010) 1077--1098.

\bibitem{fei2020unified}
F.~Fei, J.~Zhang, J.~Li, Z.~Liu, A unified stochastic particle
  bhatnagar-gross-krook method for multiscale gas flows, Journal of
  Computational Physics 400 (2020) 108972.

\bibitem{xu-book}
K.~Xu, Direct Modeling for Computational Fluid Dynamics: Construction and
  Application of Unified Gas-kinetic Scheme, World Scientic, 2015.

\bibitem{sun2015asymptotic}
W.~Sun, S.~Jiang, K.~Xu, S.~Li, An asymptotic preserving unified gas kinetic
  scheme for frequency-dependent radiative transfer equations, Journal of
  Computational Physics 302 (2015) 222--238.

\bibitem{sun2017}
W.~Sun, S.~Jiang, K.~Xu, A multidimensional unified gas-kinetic scheme for
  radiative transfer equations on unstructured mesh, Journal of Computational
  Physics 351 (2017) 455--472.

\bibitem{sun2018}
W.~Sun, S.~Jiang, K.~Xu, An asymptotic preserving implicit unified gas kinetic
  scheme for requency-dependent radiative transfer equations., International
  Journal of Numerical Analysis and Modeling 15 (2018) 134--153.

\bibitem{li2020unified}
W.~Li, C.~Liu, Y.~Zhu, J.~Zhang, K.~Xu, Unified gas-kinetic wave-particle
  methods {{III}}: Multiscale photon transport, Journal of Computational
  Physics (2020) 109280.

\bibitem{liu2017}
C.~Liu, K.~Xu, A unified gas kinetic scheme for continuum and rarefied flows
  {V}: Multiscale and multi-component plasma transport, Communications in
  Computational Physics 22~(5) (2017) 1175--1223.

\bibitem{liu2019unified}
C.~Liu, Z.~Wang, K.~Xu, A unified gas-kinetic scheme for continuum and rarefied
  flows {VI}: Dilute disperse gas-particle multiphase system, Journal of
  Computational Physics 386 (2019) 264--295.

\bibitem{liu2020unified}
C.~Liu, Y.~Zhu, K.~Xu, Unified gas-kinetic wave-particle methods {{I}}:
  Continuum and rarefied gas flow, Journal of Computational Physics 401 (2020)
  108977.

\bibitem{crestetto2019asymptotically}
A.~Crestetto, N.~Crouseilles, G.~Dimarco, M.~Lemou, Asymptotically complexity
  diminishing schemes {(ACDS)} for kinetic equations in the diffusive scaling,
  Journal of Computational Physics 394 (2019) 243--262.

\bibitem{guo2015discrete}
Z.~Guo, R.~Wang, K.~Xu, Discrete unified gas kinetic scheme for all knudsen
  number flows. {II}. thermal compressible case, Physical Review E 91~(3)
  (2015) 033313.

\bibitem{zhu2019application}
L.~Zhu, Z.~Guo, Application of discrete unified gas kinetic scheme to thermally
  induced nonequilibrium flows, Computers \& Fluids 193 (2019) 103613.

\bibitem{wang2018oscillatory}
P.~Wang, W.~Su, Y.~Zhang, Oscillatory rarefied gas flow inside a three
  dimensional rectangular cavity, Physics of Fluids 30~(10) (2018) 102002.

\bibitem{zhang2018discrete}
Y.~Zhang, L.~Zhu, R.~Wang, Z.~Guo, Discrete unified gas kinetic scheme for all
  knudsen number flows. {{III}}. binary gas mixtures of maxwell molecules,
  Physical Review E 97~(5) (2018) 053306.

\bibitem{tao2018combined}
S.~Tao, H.~Zhang, Z.~Guo, L.-P. Wang, A combined immersed boundary and discrete
  unified gas kinetic scheme for particle--fluid flows, Journal of
  Computational Physics 375 (2018) 498--518.

\bibitem{zhang2017unified}
C.~Zhang, Z.~Guo, S.~Chen, Unified implicit kinetic scheme for steady
  multiscale heat transfer based on the phonon boltzmann transport equation,
  Physical Review E 96~(6) (2017) 063311.

\bibitem{luo2018multiscale}
X.-P. Luo, C.-H. Wang, Y.~Zhang, H.-L. Yi, H.-P. Tan, Multiscale solutions of
  radiative heat transfer by the discrete unified gas kinetic scheme, Physical
  Review E 97~(6) (2018) 063302.

\bibitem{su2020fast}
W.~Su, L.~Zhu, L.~Wu, Fast convergence and asymptotic preserving of the general
  synthetic iterative scheme, arXiv preprint arXiv:2003.09958.

\bibitem{zhu2020general}
L.~Zhu, X.~Pi, W.~Su, Z.-H. Li, Y.~Zhang, L.~Wu, General synthetic iteration
  scheme for non-linear gas kinetic simulation of multi-scale rarefied gas
  flows, arXiv preprint arXiv:2004.10530.

\bibitem{BGK1954}
P.~Bhatnagar, E.~Gross, M.Krook, A model for collision processes in gases.
  \uppercase\expandafter{\romannumeral1}. small amplitude processes in charged
  and neutral one-component systems, Phys. Rev. 94~(3) (1954) 511--525.

\bibitem{sharipov2012rarefied}
F.~Sharipov, I.~A. Graur, Rarefied gas flow through a zigzag channel, Vacuum
  86~(11) (2012) 1778--1782.

\bibitem{liu2016}
C.~Liu, K.~Xu, Q.~Sun, Q.~Cai, A unified gas-kinetic scheme for continuum and
  rarefied flows {IV}: Full {Boltzmann} and model equations, Journal of
  Computational Physics 314 (2016) 305--340.

\bibitem{chen2015comparative}
S.~Chen, K.~Xu, A comparative study of an asymptotic preserving scheme and
  unified gas-kinetic scheme in continuum flow limit, Journal of Computational
  Physics 288 (2015) 52--65.

\bibitem{huang2013}
J.-C. Huang, K.~Xu, P.~Yu, A unified gas-kinetic scheme for continuum and
  rarefied flows {III}: Microflow simulations, Communications in Computational
  Physics 14~(5) (2013) 1147--1173.

\bibitem{sheng2014simulation}
Q.~Sheng, G.-H. Tang, X.-J. Gu, D.~R. Emerson, Y.-H. Zhang, Simulation of
  thermal transpiration flow using a high-order moment method, International
  Journal of Modern Physics C 25~(11) (2014) 1450061.

\bibitem{zhu2017implicit}
Y.~Zhu, C.~Zhong, K.~Xu, Unified gas-kinetic scheme with multigrid convergence
  for rarefied flow study, Physics of Fluids 29~(9) (2017) 096102.

\bibitem{zhu2016implicit}
Y.~Zhu, C.~Zhong, K.~Xu, Implicit unified gas-kinetic scheme for steady state
  solutions in all flow regimes, Journal of Computational Physics 315 (2016)
  16--38.

\end{thebibliography}
\bibliographystyle{elsarticle-num}
\biboptions{numbers,sort&compress}
\end{document}